\documentclass[lettersize,journal]{IEEEtran}
\usepackage{amsmath,amsfonts}
\usepackage{algorithmic}
\usepackage{algorithm}
\usepackage{array}
\usepackage[caption=false,font=footnotesize]{subfig}
\usepackage{textcomp}
\usepackage{stfloats}
\usepackage{url}
\usepackage{verbatim}
\usepackage{graphicx}

\usepackage{cite}        
\usepackage{amsmath}     
\usepackage{amsthm}
\usepackage{amssymb}     
\usepackage{amsfonts}    
\usepackage{bm}          

\usepackage{algorithm}
\usepackage{algorithmic} 

\usepackage[most]{tcolorbox}

\usepackage{booktabs}    
\usepackage{multirow}    
\usepackage{array}       

\usepackage{float}
\usepackage{stfloats}
\usepackage{lipsum} 
\usepackage{multicol}
\usepackage{multirow}
\usepackage{graphicx}

\usepackage{textcomp}
\usepackage{booktabs}
\usepackage{bm}
\usepackage[export]{adjustbox}

\usepackage{booktabs}
\usepackage{multicol}
\usepackage{multirow}
\usepackage{gensymb}
\usepackage{ragged2e}
\usepackage{tabularx}
\usepackage{makecell}

\newtheorem{proposition}{Proposition}

\newtheorem{remark}{Remark} 

\begin{document}

	\title{\huge Tri-Hybrid Beamforming for T-RIS-Enabled Base Station}
	
	\author{Hongtao~Zhang,~\IEEEmembership{Senior~Member,~IEEE,}
		Chenlong~Ding,~\IEEEmembership{Member,~IEEE}
		
		\thanks{H.~Zhang and C.~Ding are with Beijing University of Posts and Telecommunications, Beijing 100876, China (e-mail: htzhang@bupt.edu.cn; dcl@bupt.edu.cn)}

	}

	\markboth{}
	{Shell \MakeLowercase{\textit{et al.}}: Bare Demo of IEEEtran.cls for IEEE Communications Society Journals}
	
	\maketitle

	\begin{abstract}
	Transmissive reconfigurable intelligent surfaces (T-RISs) integrated into the transmitter provide a viable realization of tri-hybrid multiple-input multiple-output (MIMO), where spatial processing is distributed across the digital, analog radio-frequency (RF), and electromagnetic (EM) domains. However, unlike conventional RIS-assisted links, a transmitter-native T-RIS directly participates in radiation formation, making T-RIS front-end modeling and weighted sum-rate (WSR)-oriented joint precoding challenging under practical hardware constraints. 
	This paper develops a unified modeling and precoding framework for transmitter-native T-RIS tri-hybrid multi-user (MU) downlink transmission. 
	Specifically, starting from a continuous-field description, the received field is characterized by the interaction among the feed array, the programmable T-RIS aperture, and the user-side propagation, leading to a cascaded baseband input-output model for MU precoding. 
	Furthermore, the same representation is instantiated in the Fresnel, Fraunhofer, and mixed-field regimes, so that near-field focusing and far-field angular steering can be handled within one front-end model. 
	Additionally, a WSR maximization problem is formulated over the digital precoder, analog network, and T-RIS coefficients under power and quantization constraints. 
	A two-level solver is then developed by coupling outer weighted minimum mean-square error (WMMSE) updates with WMMSE-induced aperture-field shaping and hardware projection. 
	Simulations validate the modeling accuracy and convergence, and show that the proposed full tri-hybrid design improves WSR over baselines while suppressing mixed-field cross-regime leakage.

	\end{abstract}
	
	
	\IEEEpeerreviewmaketitle
	
	\section{Introduction}
	
	Electrically large multiple-input multiple-output (MIMO) arrays are expected to play an important role in future wireless systems, since enlarging the effective aperture can improve beamforming gain and spatial resolution \cite{Cui2023NearField}. As the aperture continues to grow, however, the channel response is no longer fully described by conventional far-field array abstractions, and continuous-aperture propagation effects become increasingly important \cite{Gong2024Holographic}. At the same time, increasing the number of radio frequency (RF) chains and analog components in proportion to the aperture size is difficult in practice because of the associated power consumption, hardware cost, and calibration overhead.
	
	Programmable electromagnetic (EM) front ends provide a possible way to relieve this tension. Reconfigurable massive MIMO has shown that part of the spatial processing can be shifted from conventional RF-domain circuitry to the EM domain \cite{Ying2024ReconfigurableMassive}. The recently proposed tri-hybrid MIMO architecture further organizes this idea into a three-layer structure consisting of digital, analog, and EM-domain processing \cite{Heath2026TriHybrid}. Reconfigurable-antenna-based tri-hybrid designs also suggest that low-power EM control can provide useful aperture flexibility without requiring a proportional increase in RF chains \cite{Castellanos2026TriHybrid}. In this context, transmitter-side transmissive reconfigurable intelligent surfaces (T-RISs) offer a particularly attractive realization, because the transmissive surface can be integrated into the base-station front end and directly used as a programmable radiating aperture \cite{Liu2026SphericalTRISBS}.
	
	Unlike an environmental reflective RIS \cite{He2026RobustBeamforming} used as an auxiliary scatterer, the T-RIS is integrated into the base station transmitter front end. It is illuminated by a compact feed array and acts as a programmable radiating aperture. Hence, the effective downlink response is jointly determined by the baseband precoder, the analog feed network, the aperture illumination, the transmissive coefficients, and the aperture-to-user propagation. This front-end coupling is difficult to capture with conventional RIS-assisted channel models and cannot be fully exploited through separate digital, analog, or surface-domain optimization. A dedicated transmitter-native T-RIS model and a joint tri-hybrid precoding design are therefore needed for multi-user (MU) communication under practical hardware constraints.
	
	\vspace{-0.3cm}
	\subsection{Related Works}
	
	RIS technologies have been widely studied as programmable surfaces for reshaping wireless propagation \cite{Liu2024MetaAtom} and improving coverage \cite{Wu2024IntelligentSurfaces}. Recent deployment-oriented studies have further considered indoor multi-metasurface placement under blockage \cite{Liu2025CompatibleMetasurfaces}, aerial RIS-assisted multi-cell transmission \cite{Liu2024DynamicAerialRIS}, and rotatable IRS coverage with experimental verification \cite{Liu2025MaxMinCoverage}. These works provide useful insights into propagation control and deployment, but the surface is mainly treated as an external propagation modifier rather than as part of the transmitter radiation process.
	
	Transmissive and omni-directional metasurface systems have also been explored for transparent transmission and dual-side coverage. Urban intelligent metasurfaces have been used to support transmission through blocked environments \cite{Liu2025UrbanIMS}. Wall-embedded dynamic IOSs further show how transmissive surfaces may be scheduled and deployed inside indoor structures \cite{Liu2026WallIOS}. IOS-aided cell-free networks reveal that coupled transmission/reflection phase shifts can strongly affect cooperative beamforming \cite{Zhu2024IOSCellFree}. These studies are closely related to transmissive surface control, but they mainly focus on propagation enhancement rather than on MU WSR-oriented precoding for an air-fed T-RIS integrated into the transmitter front end.
	
	The emerging tri-hybrid literature is closer to the present work because it explicitly introduces an EM-domain control layer in addition to digital and analog processing. Recent studies have considered MU precoding with reconfigurable antennas \cite{Zheng2025TriHybridERA}. Tri-timescale beamforming has also been investigated to reduce the configuration burden across different control layers \cite{Liu2025TriTimescale}. Alternative front ends, such as radiation-center-antenna \cite{Li2026RCRAATriHybrid} and pinching-antenna \cite{Chen2026Pinching} assisted architectures, further indicate that tri-hybrid control is becoming a practical communication design direction. 
	
	When a T-RIS is deployed as part of base station \cite{Liu2026SphericalTRISBS, Li2025TMATRTC}, its effective downlink response is shaped by feed-to-aperture illumination, programmable transmission over the aperture, and aperture-to-user field propagation. Existing RIS-assisted and tri-hybrid MIMO models do not fully account for these coupled effects, leaving the following issues open:

	\begin{itemize}
		\item \emph{A physically meaningful yet tractable T-RIS front-end model is still lacking.}
		Existing T-RIS prototypes \cite{Tang2023TRISPrototype} and transceiver \cite{Li2024TRISTransceiver} studies mainly validate the feasibility of transmitter-side transmissive surfaces, while EM-oriented RIS models reveal the dependence on incidence, local response, and observation direction \cite{Mi2024AnalyticalEM,Chen2025AngleSensitiveRIS}. 
		However, these works do not convert the feed illumination, transmissive aperture modulation, and aperture-to-user propagation into a tractable MU input-output model for precoding design.
%
		
		\item \emph{The near-/far-field transition has not been incorporated into T-RIS aided tri-hybrid design.}
		Near-/far-field studies show that large apertures change the power scaling \cite{Yi2025NearFarIRS}, spatial degrees of freedom \cite{Liu2023NearFieldTutorial}, and beamforming behavior of wireless links \cite{Liu2025NearFieldDifferent}. 
		What remains unclear is how the T-RIS front-end factorization should be instantiated in the near, far, and mixed-field regimes, and how the digital, analog, and EM-domain variables interact when near-field focusing and far-field angular steering coexist.

		\item \emph{WSR-oriented tri-hybrid precoding is not available for the T-RIS aperture field.}
		Classical WMMSE methods optimize digital precoders over a given channel \cite{Zhao2023RethinkingWMMSE}, and RIS-aided WMMSE designs usually operate on abstract cascaded channels \cite{Choi2024WMMSE}. 
		Recent tri-hybrid beamforming methods further optimize digital, analog, and EM-domain variables for other reconfigurable front ends \cite{Zheng2025TriHybridERA,Chen2026Pinching}. 
		They do not show how the MU WSR objective should be mapped to the T-RIS aperture field and then projected onto practical analog and surface configurations.

	
	\end{itemize}
	
	\begin{figure}[t]
		\centering
		\includegraphics[width=\linewidth]{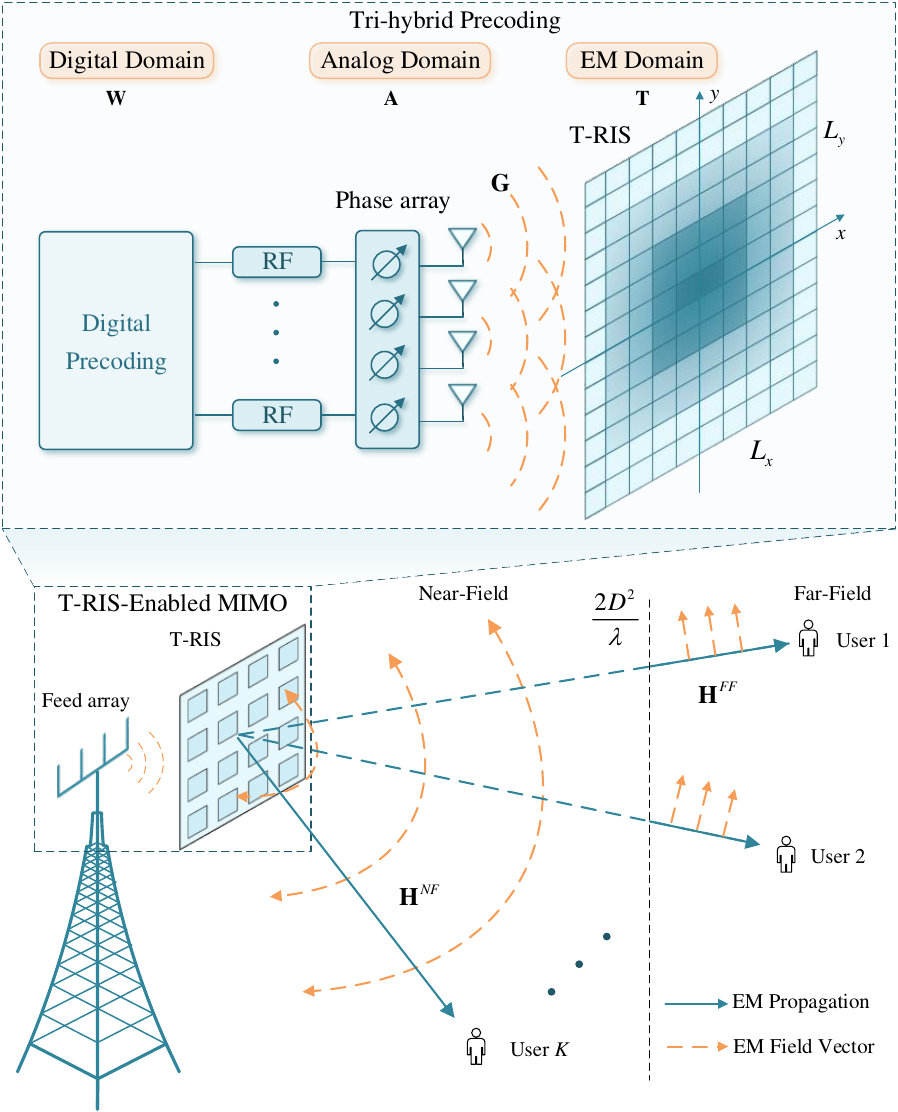}
		\caption{T-RIS-aided tri-hybrid MU downlink architecture. A compact feed array, an analog RF network, and a programmable transmissive aperture jointly shape the transmitted wavefront across the digital, analog, and EM domains.}
		\label{fig:system_model}
	\end{figure}
	
	\subsection{Contributions}
	
	This paper addresses the above gaps by linking the feed-illuminated T-RIS front end with MU WSR-oriented tri-hybrid precoding. 
	The proposed framework models the aperture radiation process, accommodates near-, far-, and mixed-field propagation, and jointly optimizes the digital, analog, and aperture-domain variables under practical hardware constraints. 
	The main contributions are as follows:
	\begin{itemize}
		\item \textit{A transmitter-native T-RIS front-end model is built from the aperture radiation process rather than from an abstract cascaded channel.}
		The feed illumination, local transmissive modulation, and aperture-to-user propagation are described in a continuous-field form. 
		The resulting model is then converted into a cascaded baseband input-output relation for MU communication design. 
		This representation preserves the main physical dependencies of the T-RIS aperture while keeping the model tractable for precoding.
		
		\item \textit{Near-, far-, and mixed-field MU transmission are unified within T-RIS front-end representation.}
		The aperture-to-user operator is specialized to the Fresnel and Fraunhofer regimes, showing how range-dependent focusing and angular steering arise from the same T-RIS front-end model. 
		The formulation is further extended to mixed-field user configurations, where near-field and far-field users share the same digital, analog, and aperture-domain variables. 
		Based on this model, a WSR maximization problem is formulated under transmit-power, quantization, and feasible-state constraints.
		
		\item \textit{A WSR-driven tri-hybrid solver is developed via WMMSE-induced aperture-field updates.}
		The nonconvex problem is handled by a two-level procedure. 
		The outer layer updates the digital precoder through WMMSE-type iterations, while the inner layer translates the communication objective into hardware-projected aperture-field updates for the analog network and T-RIS coefficients. 
		Simulations results demonstrate that the tri-hybrid precoding algorithm exhibits stable convergence and consistently outperforms the general alternating optimization algorithm and the ablation baselines in terms of weighted sum-rate (WSR) gain and low complexity.
	\end{itemize}

	\section{System Model}\label{sec:front_end}
	\vspace{-0.1cm}
	
	We consider narrowband downlink transmission at carrier wavelength $\lambda$ and wavenumber $k_0=2\pi/\lambda$. As shown in Fig.~\ref{fig:system_model}, the T-RIS aided tri-hybrid MIMO base station employs $R$ RF chains and serves $K$ single-antenna users located at $\{\mathbf u_k\in\mathbb R^3\}_{k=1}^{K}$.
	The T-RIS is a rectangular planar aperture $\mathcal S$ located on the $xy$-plane, i.e., $z=0$, and centered at the origin.
	Let its side lengths be $L_x$ and $L_y$, respectively, and define the characteristic aperture size
	$
	D \triangleq \max\{L_x,L_y\}.
	$
	The aperture is discretized into $N$ physical unit cells with centers $\{\mathbf r_n=[x_n,y_n,0]^{\mathsf T}\}_{n=1}^N$ and cell areas $\{\Delta s_n\}_{n=1}^N$.
	\(d_s\) denotes the aperture sampling pitch, with
	\(N=N_xN_y\) and \(\Delta s_n=d_s^2\).
	A compact feed array with $M$ radiating elements is placed at locations $\{\mathbf p_m\}_{m=1}^{M}$, typically on the source side of the T-RIS.
	
	Let $\mathbf s\in\mathbb C^{K\times 1}$ denote the data-symbol vector with
	$
	\mathbb E[\mathbf s\mathbf s^{\mathsf H}]=\mathbf I_K.
	$
	The digital precoder is denoted by $\mathbf W\in\mathbb C^{R\times K}$ and the analog RF network by $\mathbf A\in\mathbb C^{M\times R}$.
	The resulting feed excitation vector is
	\begin{equation}\label{eq:xf}
		\mathbf x_f=\mathbf A\mathbf W\mathbf s \in\mathbb C^{M\times 1}.
	\end{equation}
	
	The transmit-power constraint is
	\begin{equation}\label{eq:power}
		\mathbb E\!\left[\|\mathbf x_f\|_2^2\right]
		=\|\mathbf A\mathbf W\|_F^2
		\le P_{\max}.
	\end{equation}
	
	Each T-RIS cell is configured through a local transmissive coefficient.
	The local coefficient is written as
	$
	t_n, n=1,\ldots,N.
	$
	Stacking all coefficients yields the diagonal T-RIS operator
	\begin{equation}\label{eq:T_diag}
		\mathbf T \triangleq \mathrm{diag}(t_1,\ldots,t_N)\in\mathbb C^{N\times N}.
	\end{equation}
	
	\vspace{-0.6cm}
	\subsection{Feed Illumination Modeling}\label{subsec:illumination}
	
	The first building block of the T-RIS model is the incident field generated by the feed array.
	Unlike environmental RIS settings that often assume an externally given channel or a plane-wave illumination, the T-RIS considered here is directly illuminated by a compact feed array and therefore requires an explicit \emph{feed-to-aperture} description.
	
	Let $\mathbf r_s=[x_s,y_s,0]^{\mathsf T}\in\mathcal S$ denote a generic point on the aperture.
	The incident field on the T-RIS is modeled as the superposition of the contributions from all feeds:
	\begin{equation}\label{eq:Einc_cont}
		E^{\rm inc}(\mathbf r_s)
		=
		\sum_{m=1}^{M} x_{f,m}\, g_f(\mathbf r_s;\mathbf p_m),
	\end{equation}
	where $x_{f,m}$ is the excitation of feed $m$ and $g_f(\mathbf r_s;\mathbf p_m)$ denotes the illumination kernel from feed $m$ to aperture point $\mathbf r_s$.
	A physically meaningful form is
	\begin{equation}\label{eq:gf_param}
		g_f(\mathbf r_s;\mathbf p_m)
		=
		\sqrt{G_m\!\big(\theta_m(\mathbf r_s),\phi_m(\mathbf r_s)\big)}
		\frac{e^{-jk_0 R_m(\mathbf r_s)}}{4\pi R_m(\mathbf r_s)},
	\end{equation}
	where
	$
	R_m(\mathbf r_s)\triangleq \|\mathbf r_s-\mathbf p_m\|_2
	$
	is the feed-to-aperture distance,
	$G_m\!\big(\theta_m(\mathbf r_s),\phi_m(\mathbf r_s)\big)$ is the feed radiation pattern toward $\mathbf r_s$.
	
	\vspace{-0.1cm}
	\subsection{Continuous-Field of T-RIS and Equivalent Receive Model}\label{subsec:continuous_field}
	
	Under the scalar and propagating-wave assumptions, the received field at user location $\mathbf u_k=[x_k,y_k,z_k]^{\mathsf T}$ with $z_k>0$ is modeled by the Rayleigh-Sommerfeld (RS) diffraction integral \cite{Mi2024AnalyticalEM}
	\begin{equation}\label{eq:RS_cont}
		E(\mathbf u_k)
		=
		\iint_{\mathcal S}
		T(\mathbf r_s)\,
		E^{\rm inc}(\mathbf r_s)\,
		\kappa(\mathbf u_k,\mathbf r_s)\,
		\mathrm d s,
	\end{equation}
	where $T(\mathbf r_s)$ is the local transmissive coefficient at $\mathbf r_s$ and $\kappa(\mathbf u_k,\mathbf r_s)$ denotes the aperture-to-user propagation kernel.
	
	For a planar aperture with normal along the $+z$ direction, we use the radiative Rayleigh-Sommerfeld kernel
	\begin{equation}\label{eq:kernel_def}
		\kappa(\mathbf u,\mathbf r_s)
		\triangleq
		\frac{1}{j\lambda}\,
		\frac{z}{R(\mathbf u,\mathbf r_s)}\,
		\frac{e^{-jk_0R(\mathbf u,\mathbf r_s)}}{R(\mathbf u,\mathbf r_s)},
	\end{equation}
	where
	$
	R(\mathbf u,\mathbf r_s)\triangleq \|\mathbf u-\mathbf r_s\|_2
	$
	and $z$ is the $z$-coordinate of $\mathbf u$.
	The factor $e^{-jk_0R}/R$ represents spherical-wave propagation, while $z/R$ is the obliquity factor associated with the aperture normal.
	
	To integrate the continuous-field model into a communication-theoretic framework, we next map the user-side field to a complex baseband observation.
	For a single-antenna user under the scalar model, we write
	$\label{eq:disc_rx}
	y_k=\alpha_{\rm rx}E(\mathbf u_k)+n_k,
	n_k\sim\mathcal{CN}(0,\sigma^2),
	$
	where $\alpha_{\rm rx}$ collects the fixed receive-side response, including the selected polarization component, effective receive scaling, and baseband normalization.
	
	We now discretize the continuous aperture representation into a cascaded matrix model suitable for communication design.
	Applying midpoint quadrature over the $N$ physical cells yields 
	\begin{equation}\label{eq:sum_form}
		E(\mathbf u_k)
		\approx
		\sum_{n=1}^{N}
		\Big(
		t_n\,E^{\rm inc}(\mathbf r_n)
		\Big)\,
		\kappa(\mathbf u_k,\mathbf r_n)\,
		\Delta s_n.
	\end{equation}
	
	Define the sampled incident-field vector
	$
	\mathbf e_{\rm inc}
	\triangleq
	\begin{bmatrix}
		E^{\rm inc}(\mathbf r_1) & \cdots & E^{\rm inc}(\mathbf r_N)
	\end{bmatrix}^{\mathsf T}\in\mathbb C^{N\times 1}.
	$
	From \eqref{eq:Einc_cont}, its discrete form can be written as
	$
	\mathbf e_{\rm inc}=\mathbf G\mathbf x_f,
	$
	where the feed-illumination matrix $\mathbf G\in\mathbb C^{N\times M}$ is defined by
	$
	[\mathbf G]_{n,m}\triangleq g_f(\mathbf r_n;\mathbf p_m).
	$
	After T-RIS modulation, the transmitted aperture-field vector is
	$
	\mathbf e_{\rm tr}=\mathbf T\mathbf e_{\rm inc}=\mathbf T\mathbf G\mathbf x_f.
	$
	Next, define the aperture-to-user propagation matrix $\mathbf H\in\mathbb C^{K\times N}$ by
	\begin{equation}\label{eq:H_def}
		[\mathbf H]_{k,n}
		\triangleq
		\alpha_{\rm rx}\,
		\kappa(\mathbf u_k,\mathbf r_n)\,
		\Delta s_n.
	\end{equation}
	
	Then the user-side baseband vector becomes
	\begin{equation}\label{eq:cascaded}
		\mathbf y
		=
		\mathbf H\mathbf e_{\rm tr}+\mathbf n
		=
		\mathbf H\mathbf T\mathbf G\mathbf x_f+\mathbf n
		=
		\mathbf H\mathbf T\mathbf G\mathbf A\mathbf W\mathbf s+\mathbf n.
	\end{equation}
	
	It is convenient to define the tri-hybrid effective front-end channel
	$
	\mathbf H_{\rm eff}
	\triangleq
	\mathbf H\mathbf T\mathbf G\mathbf A
	\in\mathbb C^{K\times R},
	$
	so that
	$
	\mathbf y=\mathbf H_{\rm eff}\mathbf W\mathbf s+\mathbf n.
	$
	Eq.~\eqref{eq:cascaded} is the desired optimization-ready front-end representation.
	
	
	\begin{remark}
		The cascaded model in \eqref{eq:cascaded} is a physics-informed, optimization-ready reduced-order representation rather than a physical exact model. Its validity relies on the scalar and local-response assumptions in \eqref{eq:RS_cont}-\eqref{eq:kernel_def}, together with sufficiently fine aperture sampling so that the field and propagation kernel are locally smooth within each cell.
		The RS-based operator in (7)-(10) captures the deterministic aperture-to-user response of a transmitter-native T-RIS front end. Additional stochastic NLoS components can be incorporated into \(\mathbf H\), while the proposed optimization framework remains unchanged. We focus on the deterministic component to highlight the aperture-domain physics and near-/far-field factorization.
	\end{remark}

	
	\section{Channel Specialization and Problem Formulation}\label{sec:regime_wsr}
	
	Building on the front-end model in Section~\ref{sec:front_end}, this section shows how the same cascaded operator structure specializes to different propagation regimes and how it leads to a WSR-oriented MU downlink formulation.
	
	For user $k$, define the distance from the aperture center as
	$
	d_k \triangleq \|\mathbf u_k\|_2 .
	$
	A standard ruler for separating the radiative near-field and far-field regions is the Rayleigh distance
	$
	d_R \triangleq \frac{2D^2}{\lambda}.
	$
	
	Accordingly, we distinguish the following two regimes.
	(i) \textbf{Fraunhofer (far-field) regime:} $d_k \gg d_R$, where the wavefront over the aperture is well approximated by a locally planar phase profile.
	(ii) \textbf{Fresnel (radiative near-field) regime:} $d_k \lesssim d_R$, where the quadratic phase variation across the aperture must be retained.
	
	\subsection{Fresnel and Fraunhofer Specializations of the Propagation Operator}\label{subsec:H_specialization}
	
	We now specialize the propagation matrix $\mathbf H$ defined in \eqref{eq:H_def}.
	Let
	$
	R_{k,n}\triangleq \|\mathbf u_k-\mathbf r_n\|_2,
	$
	and define the direction unit vector
	$
	\hat{\mathbf u}_k \triangleq \frac{\mathbf u_k}{d_k}.
	$
	For electrically large apertures with $\|\mathbf r_n\|_2 \ll d_k$, the distance $R_{k,n}$ admits the second-order expansion \cite{Lu2024Tutorial}
	\begin{equation}\label{eq:R_expand}
		R_{k,n}
		\approx
		d_k
		-\hat{\mathbf u}_k^{\mathsf T}\mathbf r_n
		+\frac{\|\mathbf r_n\|_2^2-\big(\hat{\mathbf u}_k^{\mathsf T}\mathbf r_n\big)^2}{2d_k}.
	\end{equation}
	
	This expression makes the transition from a linear phase law to a quadratic phase law explicit.
	
	\subsubsection{Fraunhofer specialization}\label{subsubsec:fraunhofer_new}
	
	When $d_k \gg d_R$, the quadratic term in \eqref{eq:R_expand} becomes negligible.
	In this case,
	$
	e^{-jk_0 R_{k,n}}
	\approx
	e^{-jk_0 d_k}\,
	e^{+jk_0 \hat{\mathbf u}_k^{\mathsf T}\mathbf r_n},
	$
	and we also use the amplitude approximations
	$
	\frac{1}{R_{k,n}} \approx \frac{1}{d_k},
	\frac{z_k}{R_{k,n}} \approx \frac{z_k}{d_k},
	$
	where $z_k$ is the $z$-coordinate of $\mathbf u_k$.
	
	Then, we obtain the far-field propagation entry
	\begin{equation}\label{eq:H_ff}
		[\mathbf H^{\rm FF}]_{k,n}
		\approx
		\alpha_{\rm rx}\,
		\frac{1}{j\lambda}\,
		\frac{z_k}{d_k^2}\,
		e^{-jk_0 d_k}\,
		e^{+jk_0 \hat{\mathbf u}_k^{\mathsf T}\mathbf r_n}\,
		\Delta s_n .
	\end{equation}
	
	Hence, in the Fraunhofer regime, the propagation matrix behaves like a Fourier-type operator whose aperture phase is approximately linear in $\mathbf r_n$.
	
	\subsubsection{Fresnel specialization}\label{subsubsec:fresnel_new}
	
	When $d_k \lesssim d_R$, the quadratic term in \eqref{eq:R_expand} must be retained.
	Substituting \eqref{eq:R_expand} into the spherical phase term yields
	\begin{align}\label{eq:nf_phase}
		e^{-jk_0 R_{k,n}}
		\approx\;
		&e^{-jk_0 d_k}\,
		e^{+jk_0 \hat{\mathbf u}_k^{\mathsf T}\mathbf r_n} \nonumber\\
		&\times
		\exp\!\left(
		-j\frac{k_0}{2d_k}
		\Big[
		\|\mathbf r_n\|_2^2
		-
		\big(\hat{\mathbf u}_k^{\mathsf T}\mathbf r_n\big)^2
		\Big]
		\right).
	\end{align}
	
	Using again $1/R_{k,n}\approx 1/d_k$ and $z_k/R_{k,n}\approx z_k/d_k$ at the amplitude level, the Fresnel propagation entry becomes
	\begin{align}\label{eq:H_nf}
		[\mathbf H^{\rm NF}]_{k,n}
		\approx\;
		&\alpha_{\rm rx}\,
		\frac{1}{j\lambda}\,
		\frac{z_k}{d_k^2}\,
		e^{-jk_0 d_k}\,
		e^{+jk_0 \hat{\mathbf u}_k^{\mathsf T}\mathbf r_n} \nonumber\\
		&\times
		\exp\!\left(
		-j\frac{k_0}{2d_k}
		\Big[
		\|\mathbf r_n\|_2^2
		-
		\big(\hat{\mathbf u}_k^{\mathsf T}\mathbf r_n\big)^2
		\Big]
		\right)
		\Delta s_n .
	\end{align}
	
	The additional quadratic phase term is responsible for near-field focusing behavior and provides extra spatial discrimination beyond conventional far-field angular steering.
	
	Under both the Fraunhofer and Fresnel approximations, the T-RIS-aided tri-hybrid input-output relation preserves the same cascaded structure
	\begin{equation}\label{eq:unified_regime_model}
		\mathbf y
		=
		\mathbf H^{(\rho)}\mathbf T\mathbf G\mathbf A\mathbf W\mathbf s+\mathbf n,
		\qquad
		\rho\in\{{\rm FF},{\rm NF}\},
	\end{equation}
	where only the propagation matrix $\mathbf H^{(\rho)}$ changes, with $\mathbf H^{\rm FF}$ given by \eqref{eq:H_ff} and $\mathbf H^{\rm NF}$ given by \eqref{eq:H_nf}.
	
	
	Eq.~\eqref{eq:unified_regime_model} shows that the regime dependence does not alter the tri-hybrid front-end factorization itself.
	Instead, it only changes the specific realization of the propagation operator.
	Accordingly, define the regime-instantiated effective channel
	\begin{equation}\label{eq:Heff_rho}
		\mathbf H_{\rm eff}^{(\rho)}
		\triangleq
		\mathbf H^{(\rho)}\mathbf T\mathbf G\mathbf A,
		\qquad
		\rho\in\{{\rm FF},{\rm NF}\}.
	\end{equation}
	

	Importantly, the same modeling logic also extends to \emph{mixed-field} MU downlink, where different users may reside in different propagation regimes.
	Let $\mathcal K_{\rm NF}$ and $\mathcal K_{\rm FF}$ denote the near-field and far-field user sets, respectively, with $\mathcal K_{\rm NF}\cup\mathcal K_{\rm FF}=\{1,\ldots,K\}$.
	Then a mixed-field propagation operator can be written row-wise as
	\begin{equation}\label{eq:H_mix}
		\mathbf H_{\rm mix}
		\triangleq
		\begin{bmatrix}
			\mathbf H^{\rm NF}_{\mathcal K_{\rm NF}}\\
			\mathbf H^{\rm FF}_{\mathcal K_{\rm FF}}
		\end{bmatrix},
	\end{equation}
	and the corresponding effective channel is
	\begin{equation}\label{eq:Heff_mix}
		\mathbf H_{\rm eff}^{\rm mix}
		\triangleq
		\mathbf H_{\rm mix}\mathbf T\mathbf G\mathbf A.
	\end{equation}
	
	
	
	\subsection{WSR Maximization Under Practical Hardware Constraints}\label{subsec:wsr_formulation}
	
	Once the regime-specific or mixed-field propagation matrix is specified, the corresponding effective channel is fully determined.
	For notational simplicity, we use a unified symbol $\mathbf H_{\rm eff}$ in the rest of this section, with the understanding that it can stand for $\mathbf H_{\rm eff}^{\rm FF}$, $\mathbf H_{\rm eff}^{\rm NF}$, or $\mathbf H_{\rm eff}^{\rm mix}$.
	
	Let $\mathbf h_k^{\mathsf H}$ denote the $k$-th row of $\mathbf H_{\rm eff}$, and let $\mathbf w_k$ denote the $k$-th column of the digital precoder $\mathbf W$.
	Then the received signal at user $k$ is
	\begin{equation}\label{eq:yk_split}
		y_k
		=
		\mathbf h_k^{\mathsf H}\mathbf w_k s_k
		+
		\sum_{i\neq k}\mathbf h_k^{\mathsf H}\mathbf w_i s_i
		+
		n_k,
	\end{equation}
	where $n_k\sim\mathcal{CN}(0,\sigma^2)$.
	Hence, the signal-to-interference-plus-noise ratio (SINR) of user $k$ is
	\begin{equation}\label{eq:sinr_regime}
		\gamma_k
		=
		\frac{|\mathbf h_k^{\mathsf H}\mathbf w_k|^2}
		{\sum_{i\neq k}|\mathbf h_k^{\mathsf H}\mathbf w_i|^2+\sigma^2}.
	\end{equation}
	
	
		The corresponding weighted sum-rate is
	\begin{equation}\label{eq:wsr_regime}
		R_{\rm WSR}
		\triangleq
		\sum_{k=1}^{K}\mu_k \log_2\!\big(1+\gamma_k\big),
	\end{equation}
	where $\mu_k\ge 0$ is the priority weight assigned to user $k$ and $R_{\rm WSR}$ is measured in bit/s/Hz.
	Eq.~\eqref{eq:wsr_regime} is therefore not based on a new channel model separate from Section~\ref{sec:front_end}, but directly on a regime-specialized or mixed-field version of the same front-end representation.
	
	We now formulate the tri-hybrid design problem.
	The decision variables are the digital precoder $\mathbf W$, the analog network $\mathbf A$, and the T-RIS coefficients embedded in $\mathbf T$.
	
	For the analog network, we consider finite-resolution phase shifters.
	Let the feasible alphabet be
	\begin{equation}\label{eq:A_alphabet}
		\mathcal A_\beta
		\triangleq
		\left\{
		\frac{1}{\sqrt M}e^{j\frac{2\pi q}{2^\beta}}
		\,\bigg|\,
		q=0,1,\ldots,2^\beta-1
		\right\},
	\end{equation}
	where $\beta$ is the number of phase-shifter bits.
	Thus, each analog coefficient satisfies
	$
	[\mathbf A]_{m,r}\in\mathcal A_\beta,
	\qquad
	\forall\,m=1,\ldots,M,\ \forall\,r=1,\ldots,R.
	$
	
	For the T-RIS, we allow a general finite feasible set per cell \cite{Liu2025Transmitarray}:
	$
	t_n\in\mathcal T_n,
	n=1,\ldots,N,
	$
	where $\mathcal T_n$ can represent a low-resolution quantized phase set. 
	Under $b_T$-bit phase-only control with fixed modulus $\alpha_n$,
	\begin{equation}\label{eq:T_phase_only_example}
		\mathcal T_n
		=
		\left\{
		\alpha_n e^{j\frac{2\pi q}{2^{b_T}}}
		\,\bigg|\,
		q=0,1,\ldots,2^{b_T}-1
		\right\}.
	\end{equation}
	
	Combining \eqref{eq:power}, \eqref{eq:sinr_regime}, and the hardware constraints above, the WSR maximization problem is formulated as
	\begin{subequations}\label{prob:wsr_main}
		\begin{align}
			\max_{\mathbf W,\mathbf A,\mathbf T}\quad
			& \sum_{k=1}^{K}\mu_k \log_2\!\big(1+\gamma_k\big) \label{prob:wsr_main_obj}\\
			\text{s.t.}\quad
			& \|\mathbf A\mathbf W\|_F^2 \le P_{\max}, \label{prob:wsr_main_power}\\
			& [\mathbf A]_{m,r}\in\mathcal A_\beta,
			\quad \forall\,m,r, \label{prob:wsr_main_A}\\
			& t_n\in\mathcal T_n,
			\quad \forall\,n. \label{prob:wsr_main_T}
		\end{align}
	\end{subequations}
	

	\section{WSR-Driven Joint Tri-Hybrid Optimization}\label{sec:algo}
	
	%
	
	The WSR maximization in \eqref{prob:wsr_main} is difficult because the digital, analog, and EM variables are multiplicatively coupled through \(\mathbf H_{\rm eff}\), while the analog network and the T-RIS are constrained by finite-resolution or discrete hardware-feasible sets.
	We therefore adopt a two-level decomposition:
	\begin{itemize}
		\item \textbf{Outer communication-layer update:} for fixed \((\mathbf A,\mathbf T)\), update \(\mathbf W\) via WMMSE;
		\item \textbf{Inner front-end update:} for fixed \(\mathbf W\), update \((\mathbf A,\mathbf T)\) by minimizing the WMMSE-induced weighted user-plane coupling error.
	\end{itemize}
	

	
	\subsection{Outer-Layer WMMSE Update for Digital Precoding \(\mathbf W\)}\label{subsec:wmmse}
	
	For fixed \((\mathbf A,\mathbf T)\), the effective channel
	\(\mathbf H_{\rm eff}\) is fixed. We introduce a scalar linear
	equalizer \(u_k\in\mathbb C\) and a positive WMMSE weight
	\(q_k\) for each single-antenna user, and define
	$
	\hat s_k=u_k^*y_k,
	e_k\triangleq \mathbb E\!\left[|\hat s_k-s_k|^2\right].
	$
	Here, \(u_k\) is a baseband scalar equalization coefficient
	applied to the scalar received sample \(y_k\), rather than a
	spatial receive beamformer.
	Using \eqref{eq:yk_split}, we obtain
	\begin{align}\label{eq:mse_expand}
		e_k \!\!
		&=
		|u_k|^2
		\left(
		\sum_{j=1}^{K}|\mathbf h_k^{\mathsf H}\mathbf w_j|^2+\sigma^2
		\right)
		\!\!-2\Re\!\left\{u_k^* \mathbf h_k^{\mathsf H}\mathbf w_k\right\}+1,
	\end{align}
	where 
	
	The WSR maximization is equivalent to the WMMSE problem \cite{Zhao2023RethinkingWMMSE}
	\begin{equation}\label{eq:wmmse_obj}
		\min_{\mathbf W,\{u_k,q_k\}}
		\ \sum_{k=1}^{K}\mu_k\big(q_k e_k \!\!-\!\! \log q_k\big)
		\quad
		\text{s.t.}
		\|\mathbf A\mathbf W\|_F^2\le P_{\max}.
	\end{equation}
	
	For fixed \(\mathbf W\), the MMSE receive equalizer and optimal WMMSE weight are
	$
	u_k^\star
	=
	\frac{\mathbf h_k^{\mathsf H}\mathbf w_k}
	{\sum_{j=1}^{K}|\mathbf h_k^{\mathsf H}\mathbf w_j|^2+\sigma^2},
	\label{eq:u_opt}
	q_k^\star
	=
	\frac{1}{e_k^\star}.
	\label{eq:q_opt}
	$
	Define
	\begin{equation}\label{eq:Omega_def}
		\mathbf \Omega
		\triangleq
		\mathrm{diag}\!\big(
		\mu_1 q_1|u_1|^2,\ldots,\mu_K q_K|u_K|^2
		\big)\in\mathbb C^{K\times K},
	\end{equation}
	and
	\begin{equation}\label{eq:Ydes_def}
		\mathbf Y_{\mathrm{des}}
		\triangleq
		\mathrm{diag}\!\left(
		\frac{1}{u_1^*},\ldots,\frac{1}{u_K^*}
		\right)\in\mathbb C^{K\times K}.
	\end{equation}
	
	Then the digital precoder update is
	\begin{equation}\label{eq:W_closed}
		\mathbf W^\star(\lambda)
		=
		\left(
		\mathbf H_{\mathrm{eff}}^{\mathsf H}\mathbf \Omega\,\mathbf H_{\mathrm{eff}}
		+\lambda\,\mathbf A^{\mathsf H}\mathbf A
		\right)^{-1}
		\mathbf H_{\mathrm{eff}}^{\mathsf H}\mathbf \Omega\,\mathbf Y_{\mathrm{des}},
	\end{equation}
	where \(\lambda\ge 0\) is chosen such that
	$
	\|\mathbf A\mathbf W^\star(\lambda)\|_F^2\le P_{\max}.
	$
	
	The outer-layer update itself is standard.
	Its role here is to generate the communication-aware supervisory quantities \((\mathbf\Omega,\mathbf Y_{\mathrm{des}})\), which directly enter the inner front-end update.
	
	\subsection{Inner T-RIS Front-End Update for \((\mathbf A,\mathbf T)\)}\label{subsec:inner}
	
	
	The main algorithmic principle of this chapter is to preserve the \emph{exact} hardware-layer subproblem induced by WMMSE, instead of replacing it by a heuristic front-end objective.
	
	\begin{proposition}\label{prop:wmmse_inner_exact}
		For fixed \(\{u_k,q_k\}_{k=1}^{K}\) and fixed digital precoder \(\mathbf W\), optimizing \((\mathbf A,\mathbf T)\) in the WMMSE reformulation of \eqref{prob:wsr_main} is equivalent, up to additive constants independent of \((\mathbf A,\mathbf T)\), to minimizing
		\begin{equation}\label{eq:J_inner}
			J_{\mathrm{in}}(\mathbf A,\mathbf T)
			\triangleq
			\left\|
			\mathbf \Omega^{1/2}
			\big(
			\mathbf H\mathbf T\mathbf G\mathbf A\mathbf W
			-
			\mathbf Y_{\mathrm{des}}
			\big)
			\right\|_F^2,
		\end{equation}
		where \(\mathbf H\in\{\mathbf H^{\rm FF},\mathbf H^{\rm NF},\mathbf H_{\rm mix}\}\) is the instantiated propagation operator.
	\end{proposition}
	
	\begin{proof}
		The result follows from the standard WMMSE identity and completion of the square \cite{Choi2024WMMSE}.
		A full derivation is provided in Appendix~\ref{app:wmmse_ls}.
	\end{proof}
	
	Fix \(\mathbf W\) and define the user-plane coupling matrix
	$
	\mathbf Y
	\triangleq
	\mathbf H\mathbf T\mathbf G\mathbf A\mathbf W
	\in\mathbb C^{K\times K}.
	$
	Then Proposition~\ref{prop:wmmse_inner_exact} shows that the hardware update is governed by the exact communication-aware weighted least-squares (LS) surrogate in \eqref{eq:J_inner}.
	This is the key bridge from the communication layer to the front-end layer.
	
	\subsubsection{RIS aperture-field $\mathbf E_{\mathrm{tr}} = \mathbf T\mathbf G\mathbf A\mathbf W$ Update}\label{subsec:backprop}
	
	Let
	$
	\mathbf R
	\triangleq
	\mathbf H\mathbf T\mathbf G\mathbf A\mathbf W
	-
	\mathbf Y_{\mathrm{des}}
	\in\mathbb C^{K\times K},
	$
	define
	$
	\mathbf E_{\mathrm{tr}} \triangleq \mathbf T\mathbf E_{\mathrm{inc}},
	\mathbf E_{\mathrm{inc}} \triangleq \mathbf G\mathbf S,
	\mathbf S \triangleq \mathbf A\mathbf W.
	$
	Then \(\mathbf Y=\mathbf H\mathbf E_{\mathrm{tr}}\), and the Wirtinger gradient of \(J_{\mathrm{in}}\) with respect to \(\mathbf E_{\mathrm{tr}}^*\) is
	$\label{eq:grad_E}
	\nabla_{\mathbf E_{\mathrm{tr}}^*}J_{\mathrm{in}}
	=
	\mathbf H^{\mathsf H}\mathbf \Omega\,\mathbf R.
	$
	This yields the aperture-field update
	\begin{equation}\label{eq:Etr_update}
		\mathbf E_{\mathrm{tr}}^{(+)}
		\leftarrow
		\mathbf E_{\mathrm{tr}}-\eta\,\mathbf H^{\mathsf H}\mathbf \Omega\,\mathbf R,
	\end{equation}
	with step size \(\eta>0\).
	
	To analyze the continuous exact-core update in \eqref{eq:Etr_update}, define the surrogate-space objective
	$\label{eq:fE_def}
	f(\mathbf E)
	\triangleq
	\left\|
	\mathbf \Omega^{1/2}
	\big(
	\mathbf H\mathbf E-\mathbf Y_{\mathrm{des}}
	\big)
	\right\|_F^2,
	$
	whose Wirtinger gradient is
	\begin{equation}\label{eq:fE_grad}
		\nabla_{\mathbf E^*} f(\mathbf E)
		=
		\mathbf H^{\mathsf H}\mathbf \Omega
		\big(
		\mathbf H\mathbf E-\mathbf Y_{\mathrm{des}}
		\big).
	\end{equation}

		\begin{proposition}\label{prop:descent_E}
		The gradient in \eqref{eq:fE_grad} is Lipschitz continuous with constant
		\begin{equation}\label{eq:Lipschitz_E}
			L_E=\|\mathbf H^{\mathsf H}\mathbf \Omega \mathbf H\|_2.
		\end{equation}
		
		Consequently, for the unprojected update
		$\label{eq:E_unproj}
			\mathbf E^{(+)}
			=
			\mathbf E-\eta \nabla_{\mathbf E^*} f(\mathbf E),
		$
		one has
		\begin{equation}\label{eq:f_descent_quant}
			f(\mathbf E^{(+)})
			\le
			f(\mathbf E)
			-
			\eta
			\left(
			2-L_E\eta
			\right)
			\left\|
			\nabla_{\mathbf E^*} f(\mathbf E)
			\right\|_F^2 .
		\end{equation}
		
		In particular, for any \(0<\eta\le 1/L_E\),
		\begin{equation}\label{eq:f_descent_simple}
			f(\mathbf E^{(+)})
			\le
			f(\mathbf E)
			-
			\eta
			\left\|
			\nabla_{\mathbf E^*} f(\mathbf E)
			\right\|_F^2
			\le
			f(\mathbf E).
		\end{equation}
	\end{proposition}
	
	\begin{proof}
		For any two matrices \(\mathbf E_1\) and \(\mathbf E_2\),
		\begin{align}
			\left\|
			\nabla_{\mathbf E^*} f(\mathbf E_1)
			-
			\nabla_{\mathbf E^*} f(\mathbf E_2)
			\right\|_F
			&=
			\left\|
			\mathbf H^{\mathsf H}\mathbf \Omega \mathbf H
			(\mathbf E_1-\mathbf E_2)
			\right\|_F \nonumber\\
			&\le
			\|\mathbf H^{\mathsf H}\mathbf \Omega \mathbf H\|_2
			\|\mathbf E_1-\mathbf E_2\|_F,
		\end{align}
		which establishes \eqref{eq:Lipschitz_E}.
		
		Let
		$
		\mathbf G_E
		\triangleq
		\nabla_{\mathbf E^*}f(\mathbf E)
		=
		\mathbf H^{\mathsf H}\mathbf \Omega
		(\mathbf H\mathbf E-\mathbf Y_{\rm des}).
		$
		For any perturbation \(\mathbf D\), the quadratic structure of \(f(\mathbf E)\) gives
		\begin{align}
			f(\mathbf E+\mathbf D)
			=&\,
			f(\mathbf E)
			+
			2\Re\!\left\{
			\left\langle
			\mathbf G_E,\mathbf D
			\right\rangle
			\right\}
			+
			\left\|
			\mathbf \Omega^{1/2}\mathbf H\mathbf D
			\right\|_F^2 \nonumber\\
			\le&\,
			f(\mathbf E)
			+
			2\Re\!\left\{
			\left\langle
			\mathbf G_E,\mathbf D
			\right\rangle
			\right\}
			+
			L_E\|\mathbf D\|_F^2 .
		\end{align}
		
		By setting \(\mathbf D=-\eta\mathbf G_E\), we obtain
		\begin{align}
			f(\mathbf E^{(+)})
			\le&
			f(\mathbf E)
			-
			2\eta\|\mathbf G_E\|_F^2
			+
			L_E\eta^2\|\mathbf G_E\|_F^2 \nonumber\\
			=&\,
			f(\mathbf E)
			-
			\eta(2-L_E\eta)\|\mathbf G_E\|_F^2,
		\end{align}
		Equation \eqref{eq:f_descent_simple} follows directly when \(0<\eta\le 1/L_E\).
	\end{proof}

	\begin{remark}
		Proposition~\ref{prop:descent_E} applies only to the unprojected continuous surrogate step. After hardware projection on \(\mathbf T\) and \(\mathbf A\), exact surrogate monotonicity is generally no longer guaranteed and is handled by the outer-loop stabilization in Section~\ref{subsec:accept}.
	\end{remark}
	
	\label{subsubsec:ff_realization}
	For fraunhofer operator realization,
	using \eqref{eq:H_ff}, the far-field operator is Fourier-like:
	$\label{eq:ff_forward}
	(\mathbf H^{\rm FF}\mathbf E)_k
	=
	c_k^{\rm FF}
	\sum_{n=1}^{N}
	E_n\,e^{+jk_0 \hat{\mathbf u}_k^{\mathsf T}\mathbf r_n}\,\Delta s_n,
	$
	where
	$
	c_k^{\rm FF}
	\triangleq
	\alpha_{\rm rx}\,
	\frac{1}{j\lambda}\,
	\frac{z_k}{d_k^2}\,
	e^{-jk_0 d_k}.
	$
	Its adjoint is
	\begin{equation}\label{eq:ff_adjoint}
		(\mathbf H^{\rm FF})^{\mathsf H}\mathbf z
		=
		\sum_{k=1}^{K}
		(c_k^{\rm FF})^* z_k
		\mathbf a_k^{\rm FF},
	\end{equation}
	where the \(n\)-th entry of \(\mathbf a_k^{\rm FF}\) is
	$
	[\mathbf a_k^{\rm FF}]_n
	=
	e^{-jk_0 \hat{\mathbf u}_k^{\mathsf T}\mathbf r_n}\Delta s_n.
	$
	
	\label{subsubsec:nf_realization}
	For fresnel operator realization,
	using \eqref{eq:H_nf}, the near-field operator is quadratic-phase/focusing-like:
	$\label{eq:nf_forward}
	(\mathbf H^{\rm NF}\mathbf E)_k
	=
	c_k^{\rm NF}
	\sum_{n=1}^{N}
	E_n\,q_{k,n}^{\rm NF}\,\Delta s_n,
	$
	where
	$
	c_k^{\rm NF}
	\triangleq
	\alpha_{\rm rx}\,
	\frac{1}{j\lambda}\,
	\frac{z_k}{d_k^2}\,
	e^{-jk_0 d_k},
	$
	and
	$
	q_{k,n}^{\rm NF}
	\triangleq
	e^{+jk_0 \hat{\mathbf u}_k^{\mathsf T}\mathbf r_n}
	\exp\!\left(
	-j\frac{k_0}{2d_k}
	\Big[
	\|\mathbf r_n\|_2^2
	-
	\big(\hat{\mathbf u}_k^{\mathsf T}\mathbf r_n\big)^2
	\Big]
	\right).
	$
	Its adjoint is
	\begin{equation}\label{eq:nf_adjoint}
		(\mathbf H^{\rm NF})^{\mathsf H}\mathbf z
		=
		\sum_{k=1}^{K}
		(c_k^{\rm NF})^* z_k
		\mathbf a_k^{\rm NF},
	\end{equation}
	where the \(n\)-th entry of \(\mathbf a_k^{\rm NF}\) is
	$
	[\mathbf a_k^{\rm NF}]_n
	=
	(q_{k,n}^{\rm NF})^*\Delta s_n.
	$
	
	To ensure sufficient decrease of the continuous surrogate before hardware projection, we apply Armijo backtracking to the aperture-field objective
	$\label{eq:fE_armijo}
	f(\mathbf E)
	\triangleq
	\left\|
	\mathbf \Omega^{1/2}
	\big(
	\mathbf H\mathbf E-\mathbf Y_{\mathrm{des}}
	\big)
	\right\|_F^2 .
	$
	For the exact-core update, let
	$
	\mathbf G_E=\mathbf H^{\mathsf H}\mathbf \Omega\,\mathbf R.
	$
	Starting from an initial step size \(\eta_0\), we shrink \(\eta\) according to
	\begin{align}\label{eq:armijo}
		\text{while }\ 
		f(\mathbf E_{\mathrm{tr}}-\eta\mathbf G_E)
		&>
		f(\mathbf E_{\mathrm{tr}})
		-c\,\eta\,\|\mathbf G_E\|_F^2, \nonumber\\
		\eta &\leftarrow \tau_{\rm bt} \eta,
	\end{align}
	where \(0<\tau_{\rm bt}<1\) and \(0<c<1\).
	This guarantees sufficient decrease of the continuous exact-core surrogate before the hardware projection step.
	
	\subsubsection{Mixed-Field Task-Aware Extension}\label{subsec:mixed_task_ext}
	
	Although the exact surrogate in \eqref{eq:J_inner} already accommodates mixed-field users through \(\mathbf H_{\rm mix}\), in mixed-field coexistence the aggregate residual may be dominated by the stronger regime and thus obscure cross-regime leakage. To better exploit this structure while retaining the same WMMSE-induced core, we introduce a block-structured task-aware extension based on the exact mixed-field decomposition below.
	
	When \(\mathbf H=\mathbf H_{\rm mix}\), let \(\mathbf P_{\rm NF}\in\{0,1\}^{K_{\rm NF}\times K}\) and \(\mathbf P_{\rm FF}\in\{0,1\}^{K_{\rm FF}\times K}\) denote the row-selection matrices associated with the near-field and far-field user sets, respectively, where \(K_{\rm NF}=|\mathcal K_{\rm NF}|\) and \(K_{\rm FF}=|\mathcal K_{\rm FF}|\).
	Similarly, let \(\mathbf Q_{\rm NF}\in\{0,1\}^{K\times K_{\rm NF}}\) and \(\mathbf Q_{\rm FF}\in\{0,1\}^{K\times K_{\rm FF}}\) denote the stream-selection matrices that extract the streams intended for near-field and far-field users, respectively.
	Define the row-selected mixed-field propagation blocks and the corresponding WMMSE weight blocks as
	$
		\widetilde{\mathbf H}_{\rm NF}
		\triangleq
		\mathbf P_{\rm NF}\mathbf H_{\rm mix},
		\widetilde{\mathbf H}_{\rm FF}
		\triangleq
		\mathbf P_{\rm FF}\mathbf H_{\rm mix},
		\label{eq:task_H_blocks}
		\mathbf \Omega_{\rm NF}
		\triangleq
		\mathbf P_{\rm NF}\mathbf \Omega \mathbf P_{\rm NF}^{\mathsf H},
		\mathbf \Omega_{\rm FF}
		\triangleq
		\mathbf P_{\rm FF}\mathbf \Omega \mathbf P_{\rm FF}^{\mathsf H}.
		\label{eq:task_Omega_blocks}
	$
	Let \(\mathbf Y=\mathbf H_{\rm mix}\mathbf E_{\rm tr}\).
	Then define the four block components
	\begin{align}
		\mathbf Y_{\rm NN}
		&\triangleq \mathbf P_{\rm NF}\mathbf Y\mathbf Q_{\rm NF}, \quad
		\mathbf Y_{\rm NF}
		\triangleq \mathbf P_{\rm NF}\mathbf Y\mathbf Q_{\rm FF},
		\label{eq:Y_NF}\\
		\mathbf Y_{\rm FN}
		&\triangleq \mathbf P_{\rm FF}\mathbf Y\mathbf Q_{\rm NF}, \quad
		\mathbf Y_{\rm FF}
		\triangleq \mathbf P_{\rm FF}\mathbf Y\mathbf Q_{\rm FF},
		\label{eq:Y_FF}
	\end{align}
	and the desired diagonal blocks
	\begin{align}
		\mathbf Y_{{\rm des},{\rm NN}}
		&\triangleq \mathbf P_{\rm NF}\mathbf Y_{\rm des}\mathbf Q_{\rm NF},
		\label{eq:Ydes_NN}\\
		\mathbf Y_{{\rm des},{\rm FF}}
		&\triangleq \mathbf P_{\rm FF}\mathbf Y_{\rm des}\mathbf Q_{\rm FF}.
		\label{eq:Ydes_FF}
	\end{align}
	
	\begin{proposition}\label{prop:mixed_exact_split}
		When \(\mathbf H=\mathbf H_{\rm mix}\), the exact inner surrogate in \eqref{eq:J_inner} admits the decomposition
		\begin{align}\label{eq:task_exact_split}
			J_{\mathrm{in}}^{\rm mix}
			=
			&
			\underbrace{
				\left\|
				\mathbf \Omega_{\rm NF}^{1/2}
				\big(
				\mathbf Y_{\rm NN}-\mathbf Y_{{\rm des},{\rm NN}}
				\big)
				\right\|_F^2
			}_{J_{\rm NN}}
			\nonumber\\
			&+
			\underbrace{
				\left\|
				\mathbf \Omega_{\rm FF}^{1/2}
				\big(
				\mathbf Y_{\rm FF}-\mathbf Y_{{\rm des},{\rm FF}}
				\big)
				\right\|_F^2
			}_{J_{\rm FF}}
			\nonumber\\
			&+
			\underbrace{
				\left\|
				\mathbf \Omega_{\rm NF}^{1/2}\mathbf Y_{\rm NF}
				\right\|_F^2
			}_{J_{{\rm NF}\leftarrow{\rm FF}}}
			+
			\underbrace{
				\left\|
				\mathbf \Omega_{\rm FF}^{1/2}\mathbf Y_{\rm FN}
				\right\|_F^2
			}_{J_{{\rm FF}\leftarrow{\rm NF}}}.
		\end{align}
	\end{proposition}
	
	\begin{proof}
		Since \(\mathbf Y_{\rm des}\) is diagonal and each stream is associated with exactly one user, the desired mixed-field coupling matrix is block diagonal after the NF/FF row-column permutation.
		Partitioning the weighted Frobenius norm in \eqref{eq:J_inner} according to the row sets \(\mathcal K_{\rm NF},\mathcal K_{\rm FF}\) and the column sets associated with NF/FF streams directly yields \eqref{eq:task_exact_split}.
	\end{proof}
	
	Proposition~\ref{prop:mixed_exact_split} shows that the mixed-field exact surrogate is naturally composed of two desired fitting blocks and two cross-regime leakage blocks.
	Motivated by this structure, we define \emph{the task-aware mixed-field extension}
	\begin{align}\label{eq:task_obj}
		J_{\rm task}
		\triangleq &
		\alpha_{\rm NN}J_{\rm NN}
		+
		\alpha_{\rm FF}J_{\rm FF} \nonumber\\
		&+
		\beta_{{\rm NF}\leftarrow{\rm FF}}J_{{\rm NF}\leftarrow{\rm FF}}
		+
		\beta_{{\rm FF}\leftarrow{\rm NF}}J_{{\rm FF}\leftarrow{\rm NF}},
	\end{align}
	where \(\alpha_{\rm NN},\alpha_{\rm FF}\ge 1\) are task-balancing weights and \(\beta_{{\rm NF}\leftarrow{\rm FF}},\beta_{{\rm FF}\leftarrow{\rm NF}}\ge 1\) are cross-regime leakage-control weights.
	
	To adapt these weights, define the normalized block energies
	\begin{align}
		\bar j_{\rm NN}
		&\triangleq \frac{J_{\rm NN}}{\max\{1,K_{\rm NF}\}}, \quad\quad
		\bar j_{\rm FF}
		\triangleq \frac{J_{\rm FF}}{\max\{1,K_{\rm FF}\}},
		\label{eq:barj_FF}\\
		\bar j_{{\rm NF}\leftarrow{\rm FF}}
		&\triangleq \frac{J_{{\rm NF}\leftarrow{\rm FF}}}{\max\{1,K_{\rm NF}\}},
		\bar j_{{\rm FF}\leftarrow{\rm NF}}
		\triangleq \frac{J_{{\rm FF}\leftarrow{\rm NF}}}{\max\{1,K_{\rm FF}\}}.
		\label{eq:barj_FFleak}
	\end{align}
	
	Then we choose
	\begin{align}
		\alpha_{\rm NN}^{(i)}
		&=
		1+\rho_{\rm task}
		\frac{\bar j_{\rm NN}^{(i)}}
		{\bar j_{\rm NN}^{(i)}+\bar j_{\rm FF}^{(i)}},
		\label{eq:task_alpha_NN}\\
		\alpha_{\rm FF}^{(i)}
		&=
		1+\rho_{\rm task}
		\frac{\bar j_{\rm FF}^{(i)}}
		{\bar j_{\rm NN}^{(i)}+\bar j_{\rm FF}^{(i)}},
		\label{eq:task_alpha_FF}\\
		\beta_{{\rm NF}\leftarrow{\rm FF}}^{(i)}
		&=
		1+\rho_{\rm leak}
		\frac{\bar j_{{\rm NF}\leftarrow{\rm FF}}^{(i)}}
		{\bar j_{{\rm NF}\leftarrow{\rm FF}}^{(i)}+\bar j_{{\rm FF}\leftarrow{\rm NF}}^{(i)}},
		\label{eq:task_beta_NF}\\
		\beta_{{\rm FF}\leftarrow{\rm NF}}^{(i)}
		&=
		1+\rho_{\rm leak}
		\frac{\bar j_{{\rm FF}\leftarrow{\rm NF}}^{(i)}}
		{\bar j_{{\rm NF}\leftarrow{\rm FF}}^{(i)}+\bar j_{{\rm FF}\leftarrow{\rm NF}}^{(i)}},
		\label{eq:task_beta_FF}
	\end{align}
	where \(0\le \rho_{\rm task}\le 1\), \(0\le \rho_{\rm leak}\le 1\).
	
	When \(\rho_{\rm task}=0\) and \(\rho_{\rm leak}=0\), one has
	$
	\alpha_{\rm NN}=\alpha_{\rm FF}=\beta_{{\rm NF}\leftarrow{\rm FF}}=\beta_{{\rm FF}\leftarrow{\rm NF}}=1,
	$
	and therefore
	$\label{eq:task_reduction}
	J_{\rm task}=J_{\mathrm{in}}^{\rm mix}.
	$
	Hence, the mixed-field task-aware extension is a strict generalization of the exact-core mixed-field solver rather than a separate heuristic objective.
	For the task-aware mixed-field extension, define the block residuals
	\begin{align}
		\mathbf R_{\rm NN}
		\triangleq
		\mathbf Y_{\rm NN}-\mathbf Y_{{\rm des},{\rm NN}},
		\mathbf R_{\rm FF}
		\triangleq
		\mathbf Y_{\rm FF}-\mathbf Y_{{\rm des},{\rm FF}}.
		\label{eq:R_FF}
	\end{align}
	
	Then the corresponding aperture-field gradient is
	\begin{align}\label{eq:task_grad}
		\mathbf G_E^{\rm task}
		=
		&
		\alpha_{\rm NN}
		\widetilde{\mathbf H}_{\rm NF}^{\mathsf H}
		\mathbf \Omega_{\rm NF}
		\mathbf R_{\rm NN}
		\mathbf Q_{\rm NF}^{\mathsf H}
		+
		\alpha_{\rm FF}
		\widetilde{\mathbf H}_{\rm FF}^{\mathsf H}
		\mathbf \Omega_{\rm FF}
		\mathbf R_{\rm FF}
		\mathbf Q_{\rm FF}^{\mathsf H}
		\nonumber\\
		&+
		\beta_{{\rm NF}\leftarrow{\rm FF}}
		\widetilde{\mathbf H}_{\rm NF}^{\mathsf H}
		\mathbf \Omega_{\rm NF}
		\mathbf Y_{\rm NF}
		\mathbf Q_{\rm FF}^{\mathsf H}
		\nonumber\\
		&+
		\beta_{{\rm FF}\leftarrow{\rm NF}}
		\widetilde{\mathbf H}_{\rm FF}^{\mathsf H}
		\mathbf \Omega_{\rm FF}
		\mathbf Y_{\rm FN}
		\mathbf Q_{\rm NF}^{\mathsf H}.
	\end{align}
	
	The task-aware update is then
	\begin{equation}\label{eq:task_update}
	\mathbf E_{\mathrm{tr}}^{(+)}
	\leftarrow
	\mathbf E_{\mathrm{tr}}
	-
	\eta_{\rm task}\mathbf G_E^{\rm task}.
	\end{equation}
	
	\begin{proposition}\label{prop:mixed_task_descent}
		For fixed \(\alpha_{\rm NN},\alpha_{\rm FF},\beta_{{\rm NF}\leftarrow{\rm FF}},\beta_{{\rm FF}\leftarrow{\rm NF}}>0\), define
		\begin{align}\label{eq:task_L}
			L_{\rm task}
			\triangleq\;
			&
			\left(\alpha_{\rm NN}+\beta_{{\rm NF}\leftarrow{\rm FF}}\right)
			\left\|
			\widetilde{\mathbf H}_{\rm NF}^{\mathsf H}
			\mathbf \Omega_{\rm NF}
			\widetilde{\mathbf H}_{\rm NF}
			\right\|_2
			\nonumber\\
			&+
			\left(\alpha_{\rm FF}+\beta_{{\rm FF}\leftarrow{\rm NF}}\right)
			\left\|
			\widetilde{\mathbf H}_{\rm FF}^{\mathsf H}
			\mathbf \Omega_{\rm FF}
			\widetilde{\mathbf H}_{\rm FF}
			\right\|_2 .
		\end{align}
		
		Then, for any \(0<\eta_{\rm task}\le 1/L_{\rm task}\), the unprojected update in \eqref{eq:task_update} yields a non-increasing sequence of \(J_{\rm task}\).
	\end{proposition}
	
	\begin{proof}
		For fixed task-aware weights, \(J_{\rm task}\) is a smooth quadratic function of \(\mathbf E_{\rm tr}\), and \eqref{eq:task_L} is a valid Lipschitz bound of its gradient since \(\|\mathbf Q_{\rm NF}\|_2=\|\mathbf Q_{\rm FF}\|_2=1\). The claim then follows directly from the descent lemma.
	\end{proof}
	
	The extension in \eqref{eq:task_obj} is therefore a structured refinement of the exact-core mixed-field surrogate in \eqref{eq:task_exact_split}, where the \(\beta\)-weights mainly suppress cross-regime leakage and the \(\alpha\)-weights provide only a mild task-balancing correction.

	The task-aware extension is activated only when \(\mathbf H=\mathbf H_{\rm mix}\) and both \(\mathcal K_{\rm NF}\) and \(\mathcal K_{\rm FF}\) are nonempty, otherwise, the exact-core update in \eqref{eq:Etr_update} is used.
	
	\subsubsection{Hardware-Native Updates for \((\mathbf A,\mathbf T)\)}\label{subsec:hardware_native}
	
	Given the updated target field \(\mathbf E_{\mathrm{tr}}^{(+)}\) and the incident field samples \(\mathbf E_{\mathrm{inc}}\), each T-RIS coefficient is first fitted by a regularized least-squares (LS) estimate. Since the \(n\)-th row satisfies
	$
	\mathbf e_{{\rm tr},n}^{(+)}
	\approx
	t_n\mathbf e_{{\rm inc},n},
	$
	we solve
	$\label{eq:tn_LS_prob}
		t_n^{\mathrm{LS}}
		=
		\arg\min_{t_n}
		\left\|
		t_n\mathbf e_{{\rm inc},n}
		-
		\mathbf e_{{\rm tr},n}^{(+)}
		\right\|_2^2 .
	$
	The unconstrained solution is
	$\label{eq:tn_LS}
		t_n^{\mathrm{LS}}
		=
		\frac{
			\mathbf e_{{\rm tr},n}^{(+)}
			\mathbf e_{{\rm inc},n}^{\mathsf H}
		}{
			\|\mathbf e_{{\rm inc},n}\|_2^2
		},
	$
	where \(\mathbf e_{{\rm inc},n}\) and \(\mathbf e_{{\rm tr},n}^{(+)}\) denote the \(n\)-th rows of \(\mathbf E_{\mathrm{inc}}\) and \(\mathbf E_{\mathrm{tr}}^{(+)}\), respectively.
	
	For practical calibrated hardware, this projection is naturally implemented as a nearest-state selection over the feasible set $\mathcal T_n$. Specifically, for any $z\in\mathbb C$, we define
	$\label{eq:T_state_proj}
		\Pi_{\mathcal T_n}(z)
		\triangleq
		\arg\min_{\tau\in\mathcal T_n}
		|z-\tau|^2 .
	$
	Then the T-RIS update is given by
	\begin{equation}\label{eq:T_proj}
		t_n
		\leftarrow
		\Pi_{\mathcal T_n}
		\big(
		t_n^{\mathrm{LS}}
		\big),
		\qquad
		\mathbf T\leftarrow \mathrm{diag}(t_1,\ldots,t_N).
	\end{equation}

	
	After updating \(\mathbf T\), map the desired transmitted field back to a desired incident field via the regularized diagonal pseudo-inverse
	$\label{eq:T_pinv}
		\mathbf T^\dagger
		\triangleq
		\mathrm{diag}\!\left(
		\frac{t_1^*}{|t_1|^2},
		\ldots,
		\frac{t_N^*}{|t_N|^2}
		\right).
	$
	For unit-modulus T-RIS coefficients, \(\mathbf T^\dagger\) reduces to \(\mathbf T^{\mathsf H}\).
	Then,
	$\label{eq:Einc_des}
		\mathbf E_{\mathrm{inc}}^{\mathrm{des}}
		\approx
		\mathbf T^\dagger \mathbf E_{\mathrm{tr}}^{(+)} .
	$
	Then backpropagate this desired incident field to the feed plane by solving the regularized LS problem
	\begin{equation}\label{eq:S_des_ls}
		\mathbf S^{\mathrm{des}}
		=
		\arg\min_{\mathbf S}
		\left\|
		\mathbf G\mathbf S
		-
		\mathbf E_{\mathrm{inc}}^{\mathrm{des}}
		\right\|_F^2
		+
		\lambda_G\|\mathbf S\|_F^2,
	\end{equation}
	whose closed-form solution is
	$\label{eq:S_des}
		\mathbf S^{\mathrm{des}}
		\leftarrow
		(\mathbf G^{\mathsf H}\mathbf G+\lambda_G \mathbf I)^{-1}
		\mathbf G^{\mathsf H}\mathbf E_{\mathrm{inc}}^{\mathrm{des}},
	$
	where \(\lambda_G>0\) is a Tikhonov regularization parameter used to stabilize the inversion of the feed-illumination operator and to avoid excessive feed excitation power.
	Since \(\mathbf S=\mathbf A\mathbf W\), a regularized LS fit gives
	$
	\mathbf A^{\mathrm{LS}}
	=
	\mathbf S^{\mathrm{des}}\mathbf W^\dagger,
	$
	where
	$
	\mathbf W^\dagger
	\triangleq
	\mathbf W^{\mathsf H}
	(\mathbf W\mathbf W^{\mathsf H}+\delta \mathbf I)^{-1},
	$
	and $\delta>0$ is a small Tikhonov regularization parameter used to stabilize the right pseudo-inverse of $\mathbf W$.
	
	Finally,
	\begin{equation}\label{eq:A_proj}
		\mathbf A
		\leftarrow
		\Pi_{\mathcal A_\beta}
		\big(
		\mathbf A^{\mathrm{LS}}
		\big).
	\end{equation}
	
	For \(\beta\)-bit constant-modulus phase shifters,
	$\label{eq:A_beta_proj}
		[\Pi_{\mathcal A_\beta}(\mathbf A)]_{m,r}
		=
		\frac{1}{\sqrt M}\exp\!\big(jQ_\beta(\arg A_{m,r})\big),
	$
	where \(Q_\beta(\cdot)\) maps a continuous phase to its nearest point in the \(\beta\)-bit phase alphabet \(\{2\pi q/2^\beta\}_{q=0}^{2^\beta-1}\).
	
	\subsection{Feasibility Preservation and Accepted-Iterate Monotonicity}\label{subsec:accept}
	
	Since the continuous aperture-field update is performed in the surrogate space, it does not directly satisfy the discrete hardware constraints on \(\mathbf T\) and \(\mathbf A\). The following hardware-native updates therefore use regularized least-squares fitting and feasible-set projection as a practical approximation, rather than exact minimization of \eqref{eq:J_inner}.
	
	\subsubsection{Transmit-power feasibility scaling}\label{subsubsec:feasible_scale}
	
	After updating \((\mathbf A,\mathbf T)\) while keeping \(\mathbf W\) fixed, the power constraint may be violated because the analog gain changes.
	We therefore enforce feasibility through
	\begin{equation}\label{eq:W_scale}
		\zeta
		\leftarrow
		\min\!\left(
		1,\sqrt{
			\frac{P_{\max}}
			{\mathrm{tr}(\mathbf W^{\mathsf H}\mathbf A^{\mathsf H}\mathbf A\mathbf W)}
		}
		\right),
		\qquad
		\mathbf W\leftarrow \zeta\,\mathbf W.
	\end{equation}
	
	\subsubsection{Acceptance and damping of outer iterates}\label{subsubsec:bsum}
	
	At the beginning of each outer iteration, we store the last accepted
	hardware state as
	$
	(\mathbf A_{\rm old},\mathbf T_{\rm old})
	\leftarrow
	(\mathbf A,\mathbf T).
	$
	After the inner hardware loop, let
	\((\mathbf A^+,\mathbf T^+)\) denote the candidate hardware state, where
	\(\mathbf T^+=\mathrm{diag}(t_1^+,\ldots,t_N^+)\).
	Since the hardware projection may break the monotonic decrease of the
	continuous surrogate and may reduce the true WSR, we damp the candidate
	update by
	\begin{align}\label{eq:damp}
		\mathbf A(\alpha)
		&=
		\Pi_{\mathcal A_\beta}
		\big(
		(1-\alpha)\mathbf A_{\rm old}+\alpha \mathbf A^+
		\big), \nonumber\\
		t_n(\alpha)
		&=
		\Pi_{\mathcal T_n}
		\big(
		(1-\alpha)t_{n,{\rm old}}+\alpha t_n^+
		\big),
		\quad n=1,\ldots,N .
	\end{align}
	Starting from \(\alpha=1\), we repeatedly set
	\(\alpha\leftarrow\beta_{\rm acc}\alpha\), with
	\(0<\beta_{\rm acc}<1\), until the WSR improves.
	
	
	Hence, after feasibility scaling via \eqref{eq:W_scale} and damping via \eqref{eq:damp}, the WSR sequence over the \emph{accepted outer iterates} is monotonically non-decreasing. 
%
	
	%
	%
	
	\subsection{Overall Algorithm and Complexity Discussion}\label{subsec:overall_algo}
	
	The outer loop is terminated when the relative WSR improvement satisfies
	$ \label{eq:stop_outer}
	\frac{
		|\mathrm{WSR}^{(t)}-\mathrm{WSR}^{(t-1)}|
	}{
		\max\{1,\mathrm{WSR}^{(t-1)}\}
	}
	\le \epsilon_{\rm out},
	$
	or when the maximum number of outer iterations \(T_{\rm out}\) is reached.
	The inner loop is stopped when
	$ \label{eq:stop_inner}
	\frac{
		|J_{\mathrm{in}}^{(i)}-J_{\mathrm{in}}^{(i-1)}|
	}{
		\max\{1,J_{\mathrm{in}}^{(i-1)}\}
	}
	\le \epsilon_{\rm in},
	$
	or when the maximum number of inner iterations \(T_{\rm in}\) is reached.
	The whole algorithm is shown in Algorithm~\ref{alg:tri_hybrid}.
	
	\begin{algorithm}[t]
		\small
		\caption{Regime-Instantiated WMMSE-Driven Joint Tri-Hybrid Optimization}\label{alg:tri_hybrid}
		\begin{algorithmic}[1]
			\REQUIRE Instantiated propagation operator \(\mathbf H\in\{\mathbf H^{\rm FF},\mathbf H^{\rm NF},\mathbf H_{\rm mix}\}\), \(\mathbf G\), user weights \(\{\mu_k\}\), noise variance \(\sigma^2\), power budget \(P_{\max}\), feasible sets \(\mathcal A_\beta\) and \(\{\mathcal T_n\}\), iteration limits \(T_{\rm out},T_{\rm in}\).
			\ENSURE \(\mathbf W,\mathbf A,\mathbf T\)
			\STATE Initialize feasible \(\mathbf A\in\mathcal A_\beta\), \(\mathbf T\in\prod_{n=1}^{N}\mathcal T_n\), and \(\mathbf W\) such that \(\|\mathbf A\mathbf W\|_F^2\le P_{\max}\).
			\FOR{$t=1$ to $T_{\rm out}$}
			\STATE \(\mathbf H_{\rm eff}\leftarrow \mathbf H\mathbf T\mathbf G\mathbf A\).
			\STATE Update \(\{u_k,q_k\}\) via \eqref{eq:u_opt}.
			\STATE Update \(\mathbf W\) via \eqref{eq:W_closed}.
			\STATE Build \(\mathbf \Omega\) and \(\mathbf Y_{\mathrm{des}}\) via \eqref{eq:Omega_def}--\eqref{eq:Ydes_def}.
			\STATE Store the last accepted \((\mathbf A_{\rm old},\mathbf T_{\rm old})\leftarrow (\mathbf A,\mathbf T)\).
			\FOR{$i=1$ to $T_{\rm in}$}
			\STATE \(\mathbf S\leftarrow \mathbf A\mathbf W\), \(\mathbf E_{\mathrm{inc}}\leftarrow \mathbf G\mathbf S\), \(\mathbf E_{\mathrm{tr}}\leftarrow \mathbf T\mathbf E_{\mathrm{inc}}\).
			\IF{\(\mathbf H=\mathbf H_{\rm mix}\) and \(K_{\rm NF}>0\) and \(K_{\rm FF}>0\)}
			\STATE Compute \(\mathbf Y_{\rm NN},\mathbf Y_{\rm NF},\mathbf Y_{\rm FN},\mathbf Y_{\rm FF}\) via \eqref{eq:Y_NF}--\eqref{eq:Y_FF}.
			\STATE Update \(\alpha_{\rm NN},\alpha_{\rm FF},\beta_{{\rm NF}\leftarrow{\rm FF}},\beta_{{\rm FF}\leftarrow{\rm NF}}\) via \eqref{eq:task_alpha_NN}--\eqref{eq:task_beta_FF}.
			\STATE Compute \(\mathbf G_E\leftarrow \mathbf G_E^{\rm task}\) via \eqref{eq:task_grad}.
			\STATE Set \(\eta \leftarrow 1/L_{\rm task}\) with \(L_{\rm task}\) given by \eqref{eq:task_L}.
			\ELSE
			\STATE \(\mathbf R\leftarrow \mathbf H\mathbf E_{\mathrm{tr}}-\mathbf Y_{\mathrm{des}}\), \(\mathbf G_E\leftarrow \mathbf H^{\mathsf H}\mathbf \Omega\,\mathbf R\).
			\STATE Choose \(\eta\) via Armijo backtracking in \eqref{eq:armijo}.
			\ENDIF
			\STATE \(\mathbf E_{\mathrm{tr}}^{(+)}\leftarrow \mathbf E_{\mathrm{tr}}-\eta\mathbf G_E\).
			\STATE Update \(\mathbf T\) via \eqref{eq:T_proj}.
			\STATE Update \(\mathbf A\) via \eqref{eq:A_proj}.  
			\ENDFOR
			\STATE Enforce power feasibility via \eqref{eq:W_scale}.
			\STATE If WSR decreases, apply rollback via \eqref{eq:damp}.
			\ENDFOR
		\end{algorithmic}
	\end{algorithm}
	
	For the outer layer, the dominant cost comes from \eqref{eq:W_closed}, which scales as \(\mathcal O(R^3)\).
	For the inner layer, the dominant operations are applications of \(\mathbf H\), \(\mathbf H^{\mathsf H}\), \(\mathbf G\), and \(\mathbf G^{\mathsf H}\).
	In the Fresnel regime, \(\mathbf H^{\rm NF}\) can be implemented through FFT-based Fresnel or angular-spectrum transforms, yielding approximately \(\mathcal O(N\log N)\) per operator application.
	In the Fraunhofer regime, \(\mathbf H^{\rm FF}\) reduces to a Fourier-type operator, whose implementation is typically no more expensive and can in some cases be further simplified.
	For mixed-field operation, the additional overhead of the task-aware extension is limited to block extraction, scalar weight updates, and blockwise residual bookkeeping, which is negligible compared with the dominant operator applications.
	

	
	\section{Simulation Results}\label{sec:sim}
	
	The default parameters are summarized in Table~\ref{tab:sim_params_filled}. 
	We compare the following six schemes.  
	\emph{\textbf{1) No-RIS Hybrid:}} a hybrid
	precoding baseline without the T-RIS aperture, where the same RF chains,
	feed elements, analog network, and power budget are retained.   
	\emph{\textbf{2) Only Field Fitting:}} the same physical front end as the proposed method, but with $\mathbf\Omega=\mathbf I$ in \eqref{eq:J_inner}, i.e., without WSR-driven weighting.  
	\emph{\textbf{3) $\mathbf T$-only:}} optimize only $\mathbf T$ while fixing $\mathbf A$ to a random feasible initialization.  
	\emph{\textbf{4) $\mathbf A$-only:}} optimize only $\mathbf A$ while fixing $\mathbf T$ to an all-pass feasible initialization.  
	\emph{\textbf{5) Generic alternating optimization (AO):}} use the same outer-loop WMMSE update, but replace the inner hardware update by a generic projected block-coordinate alternating optimization over $\mathbf A$ and $\mathbf T$. This baseline follows the general alternating-optimization philosophy commonly used in tri-hybrid beamforming designs \cite{Li2026RCRAATriHybrid}.
	\emph{\textbf{6) Proposed design:}} jointly optimize $(\mathbf W,\mathbf A,\mathbf T)$ using Algorithm~\ref{alg:tri_hybrid}.    
	
	\begin{table}[t]
		\centering
		\caption{Default simulation parameters.}
		\label{tab:sim_params_filled}
		\begin{tabular}{l l}
			\hline
			Carrier frequency & $f_c=28$~GHz \\
			Noise / SNR & $\sigma^2=1$, $\mathrm{SNR}=10\log_{10}(P_{\max}/\sigma^2)$ \\
			T-RIS aperture & $32\lambda\times 32\lambda$ \\
			Sampling / cells & $d_s=\lambda/2$, $N_x=N_y=64$, $N=4096$ \\
			Rayleigh distance & $d_R=2048\lambda$ \\
			Feed array & $M=4\times4$, spacing $2\lambda$, plane $z=-6\lambda$ \\
			RF chains / users & $R=8$, $K=8$, $\mu_k=1$ \\
			User angles & $\varphi\sim[-60^\circ,60^\circ]$, $\theta\sim[20^\circ,70^\circ]$ \\
			User distance & $[1500\lambda , 10000\lambda]$ \\
			T-RIS quantization & $b_T=2$ bits, $|t_n|=1$ \\
			Analog quantization & $\beta=3$ bits, $|A_{m,r}|=1/\sqrt M$ \\
			Solver iterations & $T_{\rm out}=30$, $T_{\rm in}=8$ \\
			MC trials & $N_{\rm mc}=200$ \\
			\hline
		\end{tabular}
	\end{table}
	
	
	We first quantify the discretization error of the sampled aperture model. For a fixed $32\lambda\times32\lambda$ aperture, different sampling pitches are compared with a dense numerical reference with $d_{\rm ref}=\lambda/8$. Let $(\mathbf y_{\rm ref},R_{\rm ref})$ and $(\mathbf y,R)$ denote the received signal vector and WSR obtained from the reference and tested discretizations under the same physical realization and transmit configuration, respectively. We define
	$
	\varepsilon_{\rm y}
	\triangleq
	\mathbb E_{\rm MC}\left[
	\frac{\|\mathbf y_{\rm ref}-\mathbf y\|_2^2}
	{\|\mathbf y_{\rm ref}\|_2^2}
	\right],
	\varepsilon_{\rm WSR}
	\triangleq
	\mathbb E_{\rm MC}\left[
	\frac{|R_{\rm ref}-R|}{R_{\rm ref}}
	\right],
	$
	where $\mathbb E_{\rm MC}[\cdot]$ denotes the empirical average over Monte Carlo user realizations.
	Fig.~\ref{fig:disc_nmse} shows that both metrics exhibit a clear decreasing trend as the grid
	is refined. The default choice \(d_s=\lambda/2\) provides a suitable
	accuracy-complexity tradeoff for the subsequent simulations.

	\begin{figure}[t]
		\centering
		\includegraphics[width=0.85\columnwidth]{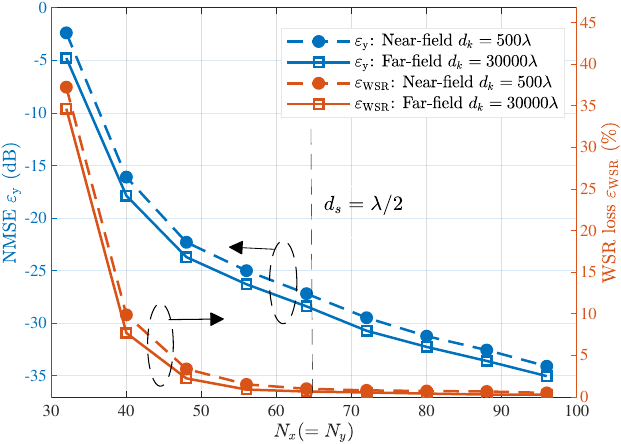}
		\caption{Validation of the continuous-to-discrete model \eqref{eq:RS_cont}-\eqref{eq:cascaded}. NMSE and induced WSR mismatch versus grid resolution for near-field and far-field users.}
		\label{fig:disc_nmse}
	\end{figure}
	
	Define the effective channel based on Rayleigh-Sommerfeld (RS) diffraction integral 
	$\mathbf H_{\rm eff}^{(\rm RS)}$, for $\rho\in\{\mathrm{NF},\mathrm{FF}\}$,  
	the model mismatch is measured by
	$
	\epsilon_H^{\rho}
	\triangleq
	\mathbb E_{\rm MC}
	\left[
	\frac{\|\mathbf H_{\rm eff }^{\rho}-\mathbf H_{\rm eff}^{\rm RS}\|_F^2}
	{\|\mathbf H_{\rm eff}^{\rm RS}\|_F^2}
	\right].
	$
	In addition, if $R_{\rho\rightarrow{\rm RS}}$ denotes the WSR obtained by designing the precoder from $\mathbf H_{\rm eff}^{(\rho)}$ and evaluating it on $\mathbf H_{\rm eff}^{(\rm RS)}$, the induced WSR loss is defined as
	$
	\Delta R_{\rho\rightarrow{\rm RS}}
	\triangleq
	100
	\frac{
		R_{{\rm RS}\rightarrow{\rm RS}}-R_{\rho\rightarrow{\rm RS}}
	}{
		R_{{\rm RS}\rightarrow{\rm RS}}
	}.
	$
	Fig.~\ref{fig:near_far_validation} shows that the Fresnel approximation remains close to the RS reference and induces negligible WSR loss over the tested range. 
	In contrast, the Fraunhofer approximation causes pronounced mismatch and WSR loss before and around $d_{\rm R}$, and becomes reliable only in the clear far-field.
	

	\begin{figure}[t]
		\vspace{-0.3cm}
		\centering
		\includegraphics[width=0.85\columnwidth]{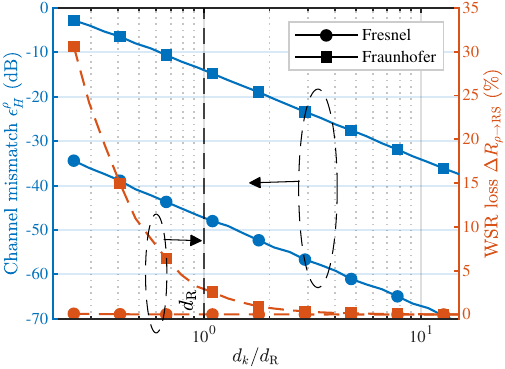}
		\caption{Validation of the Fresnel and Fraunhofer approximations against the RS reference versus the user distance.}
		\label{fig:near_far_validation}
	\end{figure}

	Fig.~\ref{fig:conv_wsr_jin} shows that the outer-loop WSR is monotone non-decreasing due to the rollback acceptance mechanism. Moreover, the proposed weighted inner update attains a higher WSR plateau than the unweighted counterpart with $\mathbf\Omega=\mathbf I$, confirming the benefit of WMMSE-aligned inner field shaping under quantized hardware constraints.
	
	\begin{figure}[t]
		\centering
		\includegraphics[width=0.775\columnwidth]{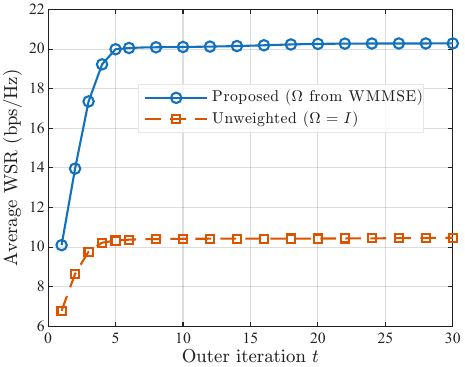}
		\caption{Convergence behavior of Algorithm~\ref{alg:tri_hybrid}. 
			}
		\label{fig:conv_wsr_jin}
	\end{figure}
	
	
	
	Fig.~\ref{fig:wsr_snr} evaluates the average WSR versus SNR in both pure near-field and mixed-field settings.  
	In both cases, the proposed full tri-hybrid design consistently outperforms all baselines. 
	The gaps to the No-RIS and only-field-fitting schemes verify the gain of the programmable aperture and WSR-aligned field shaping, while the advantage over generic AO confirms the benefit of the structured hardware update. 
	The slightly lower WSR in the mixed-field case reflects the coexistence penalty caused by heterogeneous spatial tasks sharing the same aperture and RF-chain resources.
	
	
	
	\begin{figure}[t]
		\centering
		
		\subfloat[Pure near-field users with $K_{\rm NF}=8$.]{
			\includegraphics[width=0.8\columnwidth]{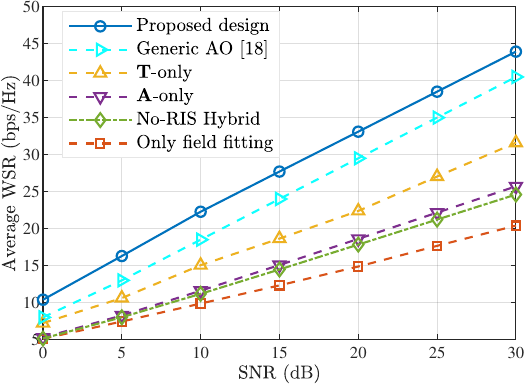}
		}\vfill
		\subfloat[Mixed-field users with $K_{\rm NF}=K_{\rm FF}=4$.]{
			\includegraphics[width=0.8\columnwidth]{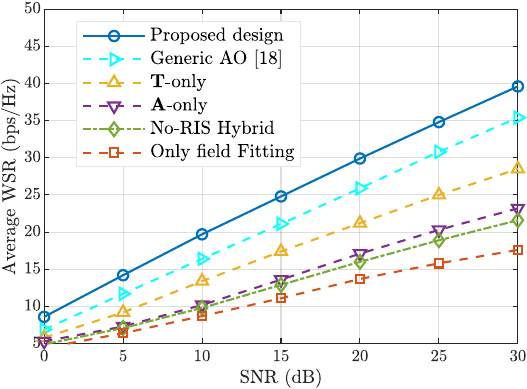}
		}
		
		\caption{Average WSR versus SNR.}
		\label{fig:wsr_snr}
		
		\vspace{-0.3cm}
	\end{figure}
	
	Fig.~\ref{fig:wsr_k} shows that proposed design remains the best scheme for all user loads and keeps a visible advantage over the generic AO baseline, especially as the system approaches the RF-limited regime. All curves gradually saturate with increasing \(K\), while the relative positions of \emph{$\mathbf T$-only} and \emph{$\mathbf A$-only} again indicate that T-RIS reconfigurability provides the main gain and analog adaptation offers a secondary improvement.
	\begin{figure}[t]
		\centering
		\includegraphics[width=0.83\columnwidth]{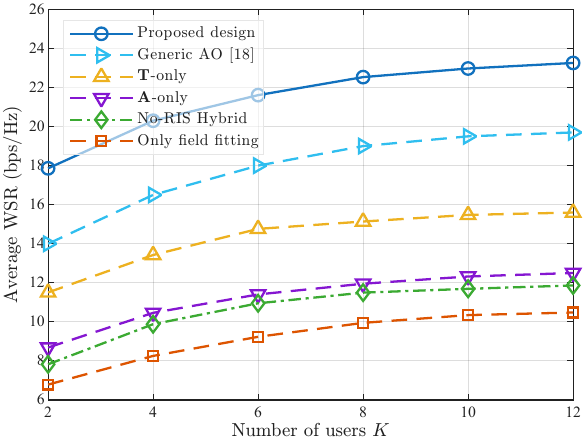}
		\caption{Average WSR versus the number of mixed-field users $K$.}
		\label{fig:wsr_k}
	\end{figure}
	
	To quantify the NF/FF coupling, we define the normalized cross-regime leakage ratio as
	$
		\frac{\|\mathbf Y_{\rm NF}\|_F^2+\|\mathbf Y_{\rm FN}\|_F^2}
		{\|\mathbf Y_{\rm NN}\|_F^2+\|\mathbf Y_{\rm FF}\|_F^2+
			\|\mathbf Y_{\rm NF}\|_F^2+\|\mathbf Y_{\rm FN}\|_F^2}.
		\label{eq:cross_regime_leakage}
	$
	Fig.~\ref{fig:mixed_field_taskaware} shows that ratio is zero at the pure-regime endpoints and increases in the mixed-field region, where near-field focusing and far-field steering compete for the same T-RIS aperture. 
	Compared with the exact-core mixed-field baseline, the task-aware extension suppresses the cross-regime leakage and consequently lifts the WSR valley around the balanced NF/FF composition. 
	
	\begin{figure}[t]
		\vspace{-0.1cm}
		\centering
		\includegraphics[width=0.9\columnwidth]{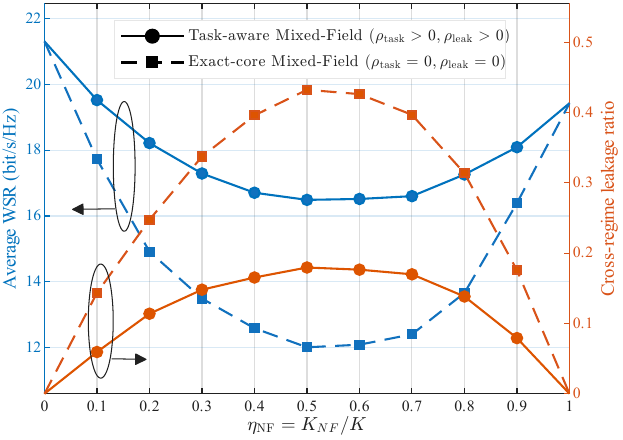}
		\caption{Verification of the mixed-field task-aware extension versus the near-field user fraction $\eta_{\rm NF}=K_{\rm NF}/K$.}
		\label{fig:mixed_field_taskaware}
		\vspace{-0.4cm}
	\end{figure}
	
	
	Fig.~\ref{fig_E1} illustrates the joint effect of the T-RIS phase resolution \(b_T\) and the analog phase resolution \(\beta\) on the average WSR. Increasing either resolution improves the performance, but the WSR varies more noticeably along the \(\beta\)-axis, indicating that the analog beamforming network remains the dominant quantization bottleneck in the considered RF-limited setting. In contrast, the gain from increasing \(b_T\) saturates earlier, implying that moderate T-RIS phase resolution is already sufficient to capture most of the aperture-side benefit.
	
	\begin{figure}[t]
		\centering
		\includegraphics[width=0.93\columnwidth]{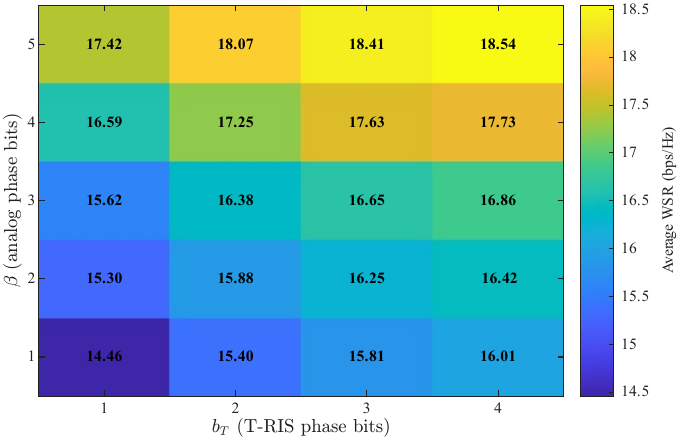}
		\caption{Quantization impact under the proposed tri-hybrid design.}
		\label{fig_E1}
	\end{figure}

	Fig.~\ref{fig:runtime_scaling} compares the average runtime per solver call of the proposed full tri-hybrid solver and the generic AO baseline. The proposed solver remains consistently faster in both dimension sweeps, and the gap becomes more pronounced as \(N\) and \(R\) increase. This indicates that the proposed structure-aware update achieves a more favorable complexity scaling than generic block-coordinate AO under the same tri-hybrid model and hardware constraints.
	
	\begin{figure}[!t]
		\vspace{-0.48cm}
		\centering
		\subfloat[]{
			\includegraphics[width=0.49\columnwidth]{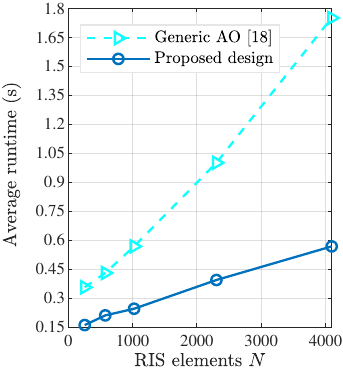}
			\label{fig:xxx_a}
 		}\hfill
		\subfloat[]{
			\hspace{-8pt}
			\includegraphics[width=0.463\columnwidth]{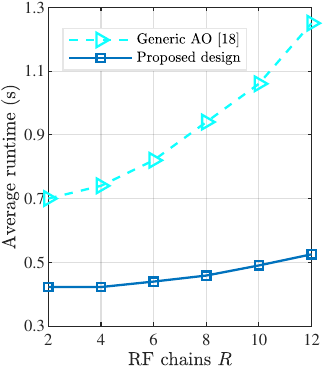}
			\label{fig:xxx_b}
		}
		\caption{(a) Average runtime versus the number of T-RIS elements $N$. (b) Average runtime versus the number of RF chains $R$ with the T-RIS size fixed.}
		\label{fig:runtime_scaling}
	\end{figure}

	\section{Conclusion}
	This paper investigated a transmitter-native T-RIS-aided tri-hybrid MU-MIMO downlink, where the digital precoder, analog RF network, and programmable transmissive aperture jointly shape the transmitted wavefront. An EM-aware framework was developed by bridging the Rayleigh-Sommerfeld continuous-field model with a cascaded baseband representation, and by unifying Fresnel and Fraunhofer operation under the same front-end factorization. To address the resulting nonconvex WSR maximization problem, a two-level algorithm was proposed, with outer-layer WMMSE digital precoding and inner-layer hardware-projected aperture-field updates for the analog network and T-RIS coefficients. Simulations verified the model accuracy, convergence behavior, and WSR gains over cost-fair and ablation baselines.
	
	\vspace{-0.3cm}
	
	\appendices
	\section{}\label{app:wmmse_ls}
	
	For user \(k\), the received signal is
	$
	y_k=\sum_{i=1}^{K}\mathbf h_k^{\mathsf H}\mathbf w_i s_i+n_k,
	n_k\sim\mathcal{CN}(0,\sigma^2),
	$
	and the scalar-equalized estimate is \(\hat s_k=u_k^*y_k\). Hence,
	\begin{align}\label{eq:app_mse}
		e_k
		=&
		\mathbb E\!\left[|\hat s_k-s_k|^2\right] \nonumber\\
		=&
		|u_k|^2
		\left(
		\sum_{j=1}^{K}|\mathbf h_k^{\mathsf H}\mathbf w_j|^2+\sigma^2
		\right)
		-2\Re\!\left\{u_k^*\mathbf h_k^{\mathsf H}\mathbf w_k\right\}
		+1 .
	\end{align}
	
	Define the user-plane coupling matrix
	$
	\mathbf Y
	\triangleq
	\mathbf H\mathbf T\mathbf G\mathbf A\mathbf W
	=
	\mathbf H_{\rm eff}\mathbf W,
	$
	whose \((k,i)\)-th entry is \(Y_{k,i}=\mathbf h_k^{\mathsf H}\mathbf w_i\). Then \eqref{eq:app_mse} becomes
	\begin{equation}\label{eq:app_mse_Y}
		e_k
		=
		|u_k|^2
		\left(
		\sum_{i=1}^{K}|Y_{k,i}|^2+\sigma^2
		\right)
		-2\Re\!\left\{u_k^*Y_{k,k}\right\}
		+1 .
	\end{equation}
	
	Multiplying by \(\mu_k q_k\) and summing over \(k\) yields
	\begin{align}\label{eq:app_weighted_sum}
		\sum_{k=1}^{K}\mu_k q_k e_k
		&=
		\sum_{k=1}^{K}\mu_k q_k |u_k|^2
		\sum_{i=1}^{K}|Y_{k,i}|^2
		-2\sum_{k=1}^{K}\mu_k q_k \Re\!\left\{u_k^*Y_{k,k}\right\}
		\nonumber\\
		&\quad
		+\sigma^2\sum_{k=1}^{K}\mu_k q_k |u_k|^2
		+\sum_{k=1}^{K}\mu_k q_k .
	\end{align}
	
	Now define
	$
	\mathbf \Omega
	\triangleq
	\mathrm{diag}\!\big(
	\mu_1q_1|u_1|^2,\ldots,\mu_Kq_K|u_K|^2
	\big),
	\mathbf Y_{\mathrm{des}}
	\triangleq
	\mathrm{diag}\!\left(
	\frac{1}{u_1^*},\ldots,\frac{1}{u_K^*}
	\right).
	$
	Then
	\begin{align}\label{eq:app_fro_expand}
		\left\|
		\mathbf \Omega^{1/2}
		(\mathbf Y-\mathbf Y_{\mathrm{des}})
		\right\|_F^2
		&=
		\sum_{k=1}^{K}\mu_k q_k |u_k|^2
		\sum_{i=1}^{K}|Y_{k,i}|^2
		\nonumber\\
		&
		-2\sum_{k=1}^{K}\mu_k q_k \Re\!\left\{u_k^*Y_{k,k}\right\}
		+\sum_{k=1}^{K}\mu_k q_k .
	\end{align}
	
	Combining \eqref{eq:app_weighted_sum} and \eqref{eq:app_fro_expand}, we obtain
	\begin{equation}\label{eq:app_equiv}
		\sum_{k=1}^{K}\mu_k q_k e_k
		=
		\left\|
		\mathbf \Omega^{1/2}
		(\mathbf Y-\mathbf Y_{\mathrm{des}})
		\right\|_F^2
		+
		\sigma^2\sum_{k=1}^{K}\mu_k q_k |u_k|^2 .
	\end{equation}
	
	Since \(\mathbf Y=\mathbf H\mathbf T\mathbf G\mathbf A\mathbf W\), the hardware-dependent part of the WMMSE objective is therefore
	\begin{equation}\label{eq:app_final}
		J_{\mathrm{in}}(\mathbf A,\mathbf T)
		=
		\left\|
		\mathbf \Omega^{1/2}
		\big(
		\mathbf H\mathbf T\mathbf G\mathbf A\mathbf W
		-
		\mathbf Y_{\mathrm{des}}
		\big)
		\right\|_F^2,
	\end{equation}
	which is exactly \eqref{eq:J_inner}. This completes the proof.

\end{document}